\documentclass{llncs}

\usepackage[all,dvips,color,arrow,curve]{xy}
\xyoption{poly}

\usepackage{amssymb,amsmath}
\usepackage{hyperref}
\usepackage{enumerate}

\spnewtheorem{observation}[theorem]{Observation}{\bfseries}{\itshape}

\renewcommand{\qed}{$\Box$}
\renewenvironment{proof}
{\noindent{\bf Proof.}\ }
{\hfill\qed\par\bigskip}

\newenvironment{proofof}[1]
{\noindent{\bf Proof of #1.}\ }
{\hfill\qed\par\bigskip}

\newcounter{claim}
	{\refstepcounter{claim}\vspace{1ex}\noindent {(\it\arabic{claim}) {#1}{}}\it}{\vspace{1ex}}
	{\noindent {}{#1}{}}{This proves~(\arabic{claim}).\vspace{1ex}}

\begin{document}

\title{Unique perfect phylogeny is $NP$-hard}

\author{Michel Habib\inst{1} and Juraj Stacho\inst{2}}

\institute{
LIAFA -- CNRS and Universit\'e Paris Diderot -- Paris VII,\\
Case 7014, 75205 Paris Cedex 13, France
(\email{habib@liafa.jussieu.fr})
\and
Caesarea Rothschild Institute, University of Haifa\\
Mt. Carmel, 31905 Haifa, Israel
(\email{stacho@cs.toronto.edu})
}

\maketitle

\begin{abstract}
We answer, in the affirmitive, the following question proposed by
Mike Steel as a \$100 challenge: {\em ``Is the following problem $NP$-hard?
Given a ternary\footnote[4]{The original formulation uses the term ``binary'',
in the sense of ``rooted binary tree'', but in this contex the two are
equivalent.} phylogenetic $X$-tree ${\cal T}$ and a collection $\cal Q$ of
quartet subtrees on $X$, is ${\cal T}$ the only tree that displays~$\cal
Q$?''}~\cite{phyl-book,steel-web}
\end{abstract}

\section{Introduction}

One of the major efforts in molecular biology has been the computation of
phylogenetic trees, or {\em phylogenies}, which describe the evolution of a set
of species from a common ancestor. A phylogenetic tree for a set of species is a
tree in which the leaves represent the species from the set and the internal
nodes represent the (hypothetical) ancestral species. One standard model for
describing the species is in terms of {\em characters}, where a character is an
equivalence relation on the species set, partitioning it into different {\em
character states}. In this model, we also assign character states to the
(hypothetical) ancestral species. The desired property is that for each state of
each character, the set of nodes in the tree having that character state forms a
connected subgraph. When a phylogeny has this property, we say it is {\em
perfect}. The Perfect Phylogeny problem \cite{gusfield} then asks {\em for a
given set of characters defining a species set, does there exist a perfect
phylogeny?} Note that we allow that states of some characters are unknown for
some species; we call such characters {\em partial}, otherwise we speak of {\em
full} characters.  This approach to constructing phylogenies has been studied
since the 1960s \cite{60s1,60s3,60s4,60s5,60s2} and was given a precise
mathematical formulation in the 1970s \cite{70s1,70s2,70s3,70s4}. In particular,
Buneman \cite{buneman} showed that the Perfect Phylogeny problem reduces to a
specific graph-theoretic problem, the problem of finding a chordal completion
of a graph that respects a prescribed colouring.  In fact, the two problems are
polynomially equivalent \cite{warnow1990}. Thus, using this formulation, it has
been proved that the Perfect Phylogeny problem is $NP$-hard in \cite{twostrikes}
and independently in \cite{steelnphard}.  These two results rely on the fact
that the input may contain partial characters. In fact, the characters in these
constructions only have two states.  If we insist on full characters, the
situation is different as for any fixed number $r$ of character states, the
problem can be solved in time polynomial \cite{agarwala1994} in the size of the
input (and exponential in $r$).  In fact, for $r=2$ (or $r=3$), the solution
exists if and only if it exists of every pair (or triple) of characters
\cite{70s4,threestate}.  Also, when the number of characters is $k$ (even if
there are partial characters), the complexity \cite{morris45} is polynomial in
the number of species (and exponential~in~$k$).

Another common formulation of this problem is the problem of a {\em consensus
tree} \cite{dekker,gordon86,steelnphard}, where a collection of subtrees with
labeled leaves is given (for instance, the leaves correspond to species of a
partial character).  Here, we ask for a (phylogenetic) tree such that each of
the input subtrees can be obtained by contracting edges from the tree (we say
that the tree {\em displays} the subtree). It turns out that the problem is
equivalent \cite{phyl-book} even if we only allow particular input subtrees, the
so-called {\em quartet trees} which have exactly six vertices and four leaves.
In fact, any ternary phylogenetic tree can be uniquely described by a collection
of quartet trees \cite{phyl-book}.  However, a collection of quartet trees does
not necessarily uniquely describe a ternary phylogenetic tree. 

This leads to a natural question: {\em  what is the complexity of deciding
whether or not a collection of quartet trees uniquely describes a (ternary)
phylogenetic tree?} This question was posed in \cite{phyl-book}, later
conjectured to be $NP$-hard and listed on M. Steel's personal webpage
\cite{steel-web} where he offers \$100 for the first proof of $NP$-hardness.  In
this paper, we answer this question by showing that the problem is indeed
$NP$-hard. In particular, we prove the following theorem.

\begin{theorem}\label{thm:unique-phyl}
It is $NP$-hard to determine, given a ternary phylogenetic $X$-tree ${\cal T}$
and a collection $\cal Q$ of quartet subtrees on $X$, whether or not ${\cal T}$
is the only phylogenetic tree that displays $\cal Q$.
\end{theorem}

We prove the theorem by describing a polynomial-time reduction from the
uniqueness problem for {\sc one-in-three-3sat}, which is $NP$-hard by the
following result of \cite{unique}. (Note that \cite{unique} gives a complete
complexity characterization of uniqueness for boolean satisfaction problems
similar to that of Shaefer \cite{shaefer}.)

\begin{theorem}{\rm\cite{unique}}\label{thm:unique-3sat}
It is $NP$-hard to decide, given an instance $I$ to {\sc one-in-three-3sat}, and
a truth assignment $\sigma$ that satisfies $I$, whether or not $\sigma$ is the
unique satisfying truth assignment for $I$.
\end{theorem}

Our construction in the reduction is essentially a modification of the
construction of \cite{twostrikes} which proves $NP$-hardness of the Perfect
Phylogeny problem. Recall that the construction of \cite{twostrikes} produces
instances ${\cal Q}$ that have a perfect phylogeny if and only if a particular
boolean formula $\varphi$ is satisfiable.  We immediately observed that these
instances ${\cal Q}$ have, in addition, the property that $\varphi$ has a unique
satisfying assignment if and only if there is a unique minimal restricted
chordal completion of the partial partition intersection graph of $\cal Q$ (for
definitions see Section \ref{sec:prelim}). This is precisely one of the two
necessary conditions for uniqueness of perfect phylogeny as proved by Semple and
Steel in \cite{semplesteel} (see Theorem \ref{thm:semplesteel}). Thus by
modifying the construction of \cite{twostrikes} to also satisfy the other
condition of uniqueness of \cite{semplesteel}, we obtained the construction that
we present in this paper. Note that, however, unlike \cite{twostrikes} which
uses {\sc 3sat}, we had to use a different $NP$-hard problem in order for the
construction to work correctly. Also, to prove that the construction is
correct, we employ a variant of the characterization of \cite{semplesteel} that
uses the more general chordal sandwich problem \cite{sandwich} instead of the
restricted chordal completion problem (see Theorem \ref{thm:sandwich}).  In
fact, by way of Theorems \ref{thm:cell-inter} and \ref{thm:sand2split}, we
establish a direct connection between the problem of perfect phylogeny and the
chordal sandwich problem, which apparently has not been yet observed.  (Note
that the connection to the (restricted) chordal completion problem of coloured
graphs as mentioned above \cite{buneman,warnow1990} is a special case of this.)
Using this result, we are able to present a much simplified
proof~of~Theorem~\ref{thm:unique-phyl}.

Finally, as a corollary, we obtain the following result.

\begin{corollary}[Chordal sandwich]
The Unique chordal sandwich problem is $NP$-hard. Counting the number of
minimal chordal sandwiches is $\#P$-complete.
\end{corollary}

The first part follows directly from Theorems \ref{thm:unique-3sat} and
\ref{thm:one}, while the second part follows from Theorem \ref{thm:one} and
\cite{counting-sat}. (Note that \cite{counting-sat} gives a complete complexity
characterization for the problem of counting satisfying assignments for boolean
satisfaction problems, just like \cite{unique} gives for uniqueness as mentioned
above).

The paper is structured as follows. First, in Section \ref{sec:prelim}, we
describe some preliminary definitions and results needed for our construction of
the reduction. In particular, we describe, based on \cite{semplesteel},
necessary and sufficient conditions for the existence of a unique perfect
phylogeny in terms of the minimal chordal sandwich problem (cf.
\cite{sandwich-strongly,sandwich}). The proof of this characterization is
postponed until Section~\ref{sec:sandwich}.  In Section~\ref{sec:constr}, we
describe the actual construction and state one of the two uniqueness conditions
(Theorem~\ref{thm:one}) relating minimal chordal sandwiches to satisfying
assignments of an instance $I$ of {\sc one-in-three-3sat}. The proof is
presented later in Section \ref{sec:proof-one}.  In
Section~\ref{sec:unique-trees}, we describe and prove the other uniqueness
condition (Theorem~\ref{thm:unique-trees}) relating satisfying assignments of
$I$ to phylogenetic trees. In Section~\ref{sec:proof-unique-phyl}, we put these
results together to prove Theorem~\ref{thm:unique-phyl}.

\section{Preliminaries}\label{sec:prelim}
We mostly follow the terminology of \cite{semplesteel,phyl-book} and
graph-theoretical notions~of~\cite{west-book}.

Let $X$ be a non-empty set.  An {\em  $X$-tree} is a pair $(T,\phi)$ where $T$
is tree and $\phi:X\rightarrow V(T)$ is a mapping such that
$\phi^{-1}(v)\neq\emptyset$ for all vertices $v\in V(T)$ of degree at most two.
An $X$-tree $(T,\phi)$ is {\em ternary} if all internal vertices of $T$ have
degree three.  Two $X$-trees $(T_1,\phi_1)$, $(T_2,\phi_2)$ are {\em isomorphic}
if there exists an isomorphism $\psi:V(T_1)\rightarrow V(T_2)$ between $T_1$ and
$T_2$ that satisfies $\phi_2=\psi\circ\phi_1$.

An $X$-tree $(T,\phi)$ is a {\em phylogenetic $X$-tree} (or a {\em free
$X$-free} in \cite{semplesteel}) if $\phi$ is bijection between $X$ and the set
of leaves of $T$. 

A {\em partial partition} of $X$ is a partition of a non-empty subset of $X$
into at least two sets. If $A_1$, $A_2$, \ldots, $A_t$ are these sets, we call
them {\em cells} of this partition, and denote the partition
$A_1|A_2|\ldots|A_t$. If $t=2$, we call the partition a {\em partial split}. A
partial split $A_1|A_2$ is trivial if $|A_1|=1$ or $|A_2|=1$.

A {\em quartet tree} is a ternary phylogenetic tree with a label set of size
four, that is, a ternary tree ${\cal T}$ with 6 vertices, 4 leaves labeled
$a,b,c,d$, and with only one non-trivial partial  split $\{a,b\}|\{c,d\}$ that
it displays.  Note that such a tree is unambiguously defined by this partial
split. Thus, in the subseqent text, we identify the quartet tree ${\cal T}$ with
the partial split $\{a,b\}|\{c,d\}$, that is, we say that $\{a,b\}|\{c,d\}$ is
both a quartet tree and a partial split.

%Note that by Lemmas 3.6, 3.7, 3.8 of \cite{semplesteel}, we can, without loss
%of generality, focus exclusively on partial splits.  In particular, the results
%listed below also hold for partial partitions, essentially by the same
%arguments.

%This is because, in most cases,
%replacing any partial partition $A_1|A_2|\ldots|A_t$ with ${t\choose 2}$ partial
%splits $A_i|A_j$ does not affect the answer.

Let ${\cal T}=(T,\phi)$ be an $X$-tree, and let $\pi=A_1|A_2|\ldots|A_t$ be a
partial partition of $X$.  We say that ${\cal T}$ {\em displays} $\pi$ if there
is a set of edges $F$ of $T$ such that, for all distinct $i,j\in\{1\ldots t\}$,  the sets
$\phi(A_i)$ and $\phi(A_j)$ are subsets of the vertex sets of different
connected components of $T-F$.  We say that an edge $e$ of $T$ is {\em
distinguished} by $\pi$ if every set of edges that displays $\pi$ in ${\cal T}$
contains $e$.

Let $\cal Q$ be a collection of partial partitions of $X$.  An $X$-tree ${\cal
T}$ {\em displays} $\cal Q$ if it displays every partial partition in $\cal Q$.
An $X$-tree ${\cal T}=(T,\phi)$ is {\em distinguished} by ${\cal Q}$ if every
internal edge of $T$ is distinguished by some partial partition in $\cal Q$; we
also say that $\cal Q$ {\em distinguishes} $\cal T$. The set $\cal Q$ {\em
defines} ${\cal T}$ if $\cal T$ displays $\cal Q$, and all other $X$-trees that
display $\cal Q$ are isomorphic to $\cal T$. Note that if $\cal Q$ defines $\cal
T$, then $\cal T$ is necessarily a ternary phylogenetic $X$-tree, since
otherwise ``resolving'' any vertex either of degree four or more, or with
multiple labels results in a non-isomorphic $X$-tree that also displays $\cal Q$
(also, see Proposition 2.6 in \cite{semplesteel}).

The {\em partial partition intersection graph} of $\cal Q$, denoted by ${\rm
int}({\cal Q})$, is a graph whose vertex set is $\{ (A,\pi)~|~{\rm
where~}$A${\rm~is~a~cell~of~}\pi\in{\cal Q}\}$ and two vertices $(A,\pi)$,
$(A',\pi')$ are adjacent just if the intersection of $A$ and $A'$ is non-empty.

A graph is {\em chordal} if it contains no induced cycle of length four or more.
A {\em chordal completion} of a graph $G=(V,E)$ is a chordal graph $G'=(V,E')$
with $E\subseteq E'$.  A {\em restricted chordal completion} of ${\rm int}({\cal
Q})$ is a chordal completion $G'$ of ${\rm int}({\cal Q})$ with the property
that if $A_1$,$A_2$ are cells of $\pi\in{\cal Q}$, then $(A_1,\pi)$ is not
adjacent to $(A_2,\pi)$ in $G'$. A restricted chordal completion $G'$ of ${\rm
int}({\cal Q})$ is {\em minimal} if no proper subgraph of $G'$ is a restricted
chordal completion of ${\rm int}({\cal Q})$.

The problem of perfect phylogeny is equivalent to the problem of determining the
existence of an $X$-tree that display the given collection ${\cal Q}$ of partial
partitions. In \cite{buneman}, it was given the following graph-theoretical
characterization.

\begin{theorem}{\rm \cite{buneman,phyl-book,steelnphard}}\label{thm:buneman}
Let ${\cal Q}$ be a set of partial partitions of a set $X$. Then there exists an
$X$-tree that displays ${\cal Q}$ if and only if there exists a restricted
chordal completion of ${\rm int}({\cal Q})$.
\end{theorem}

Of course, the $X$-tree in the above theorem might not be unique. For the
problem of uniqueness, Semple and Steel \cite{semplesteel,phyl-book} describe
necessary and sufficient conditions for when a collection of partial partitions
defines an $X$-tree.

\begin{theorem}{\rm \cite{semplesteel}}\label{thm:semplesteel}
Let ${\cal Q}$ be a collection of partial partitions of a set $X$. Let $\cal T$ be a
ternary phylogenetic $X$-tree.  Then ${\cal Q}$ defines $\cal T$ if and only if:
\vspace{-1ex}
\begin{enumerate}[(i)]
\item $\cal T$ displays ${\cal Q}$ and  is distinguished by $\cal Q$, and
\item there is a unique minimal restricted chordal completion of ${\rm
int}({\cal Q})$.
\end{enumerate}
\end{theorem}

In order to simplify our construction, we now describe a variant of the
above theorem that, instead, deals with the notion of chordal sandwich.

Let $G_1=(V,E_1)$ and $G_2=(V,E_2)$ be two graphs on the same set of vertices
with $E_1\cap E_2=\emptyset$.  A {\em chordal sandwich}\footnote[4]{In this
formulation, $E_1$ are the {\em forced} edges and $E_2$ are the {\em forbidden}
edges. See \cite{sandwich} for further details on different ways of specifying
the input to this problem.} of $(G_1$,$G_2)$
is a chordal graph $G'=(V,E')$ with $E_1\subseteq E'$ and $E'\cap
E_2=\emptyset$. A chordal sandwich $G'$ of $(G_1$,$G_2)$ is {\em minimal} if no
proper subgraph of $G'$ is a chordal sandwich of $(G_1$,$G_2)$.

The {\em cell intersection graph} of ${\cal Q}$, denoted by ${\rm int^*}({\cal
Q})$, is the graph whose vertex set is $\{A~|~{\rm
where~}$A${\rm~is~a~cell~of~}\pi\in{\cal Q}\}$ and two vertices $A$, $A'$ are
adjacent just if the intersection of $A$ and $A'$ is non-empty.  Let ${\rm
forb}({\cal Q})$ denote the graph whose vertex set is that of ${\rm int^*}({\cal
Q})$ in which there is an edge between $A$ and $A'$ just if $A$,$A'$ are cells
of some $\pi\in{\cal Q}$.

The correspondence between the partial partition intersection graph and the cell
intersection graph is captured by the following theorem.

\begin{theorem}\label{thm:cell-inter}
Let ${\cal Q}$ be a collection of partial partitions of a set $X$.  Then there
is a one-to-one correspondence between the minimal restricted chordal
completions of ${\rm int}({\cal Q})$ and the minimal chordal sandwiches of
$({\rm int^*}({\cal Q}),{\rm forb}({\cal Q}))$.

\end{theorem}
(The proof of this theorem is presented as Section \ref{sec:sandwich}.)

This combined with Theorem \ref{thm:buneman} yields that there exists a
phylogenetic $X$-tree that displays ${\cal Q}$ if and only if there exists a
chordal sandwich of $({\rm int^*}({\cal Q}),{\rm forb}({\cal Q}))$.
Conversely, we can express every instance to the chordal sandwich problem
as a corresponding instance to the problem of perfect phylogeny as follows.

\begin{theorem}\label{thm:sand2split}
Let $(G_1,G_2)$ be an instance to the chordal sandwich problem.  Then there
is a collection ${\cal Q}$ of partial splits such that there is a one-to-one
correspondence between the minimal chordal sandwiches of $(G_1,G_2)$ and the
minimal restricted chordal completions of ${\rm int}({\cal Q})$.  In particular,
there exists a chordal sandwich for $(G_1,G_2)$ if and only if there exists a
phylogenetic tree that \mbox{displays}~${\cal Q}$.
\end{theorem}
\begin{proof}
Without loss of generality, we may assume that each connected component of $G_1$
has at least three vertices. (We can safely remove any component with two or
less vertices without changing the number of minimal chordal completions, since
every such component is already chordal.)

As usual, $G_1=(V,E_1)$ and $G_2=(V,E_2)$ where $E_1\cap E_2=\emptyset$.  We
define the collection ${\cal Q}$ of partial splits (of the set $E_1$) as
follows: for every edge $xy\in E_2$, we construct the partial split $F_x|F_y$,
where $F_x$ are the edges of $E_1$ incident to $x$, and $F_y$ are the edges of
$E_1$ incident to $y$.  By definition, the vertex set of the graph ${\rm
int^*}({\cal Q})$ is precisly $\{F_v~|~v\in V\}$.  Further, it can be easily
seen that the mapping $\psi$ that, for each $v\in V$, maps $v$ to $F_v$ is an
isomorphism between $G_1$ and ${\rm int^*}({\cal Q})$.  (Here, one only needs to
verify that $F_u=F_v$ implies $u=v$; for this we use that each component of
$G_1$ has at least three vertices.) Moreover, ${\rm forb}({\cal Q})$ is
precisely $\{\psi(x)\psi(y)~|~xy\in E_2\}$ by definition.  Therefore, by
Theorem~\ref{thm:cell-inter}, there is a one-to-one correspondence between the
minimal chordal sandwiches of $(G_1,G_2)$ are the minimal restricted chordal
completions of ${\rm int}({\cal Q})$. This proves the first part of the claim;
the second part follows directly from Theorem~\ref{thm:buneman}.\end{proof}

As an immediate corollary, we obtain the following desired characterization.

\begin{theorem}\label{thm:sandwich}
Let ${\cal Q}$ be a collection of partial partitions of a set $X$. Let $\cal T$
be a ternary phylogenetic $X$-tree.  Then ${\cal Q}$ defines $\cal T$ if and
only if:
\vspace{-0.5ex}
\begin{enumerate}[(i)]
\item $\cal T$ displays ${\cal Q}$ and  is distinguished by $\cal Q$, and
\item there is a unique minimal chordal sandwich of $\Big({\rm int^*}({\cal Q})$,
${\rm forb}({\cal Q})\Big)$.
\end{enumerate}
\vspace{-2.5ex}
\end{theorem}

%o the
%Perfect Phylogeny problem in which we consider $E_1$ to be the set of species,
%$E_2$ be the set of character where each $e=xy\in E_2$ has two states $x$ and
%$y$ such that $e$ is in state $x$ in the species $f\in E_1$ just if $x$ is an
%endpoint of $f$, and is in state $y$ just if $y$ is an endpoint of $f$;
%otherwise, the state is undefined.

\section{Construction}\label{sec:constr}

Consider an instance $I$ to {\sc one-in-three-3sat}. That is, $I$ consists of
$n$ variables $v_1,\ldots,v_n$ and $m$ clauses ${\cal C}_1,\ldots, {\cal C}_m$
each of which is a disjunction of exactly three {\em literals} (i.e., variables
$v_i$ or their negations $\overline{v_i}$).

By standard arguments, we may assume that no variable appears twice in the same
clause, since otherwise we can replace the instance $I$ with an equivalent
instance with this property. In particular, we can replace each clause of the
form ${v_i}\vee \overline {v_i}\vee {v_j}$ by clauses ${v_i}\vee x\vee {v_j}$
and $\overline {v_i}\vee\overline x\vee {v_j}$ where $x$ is a new variable, and
replace each clause of the form ${v_i}\vee {v_i}\vee {v_j}$ by clauses
${v_i}\vee {v_j}\vee x$, ${v_i}\vee\overline {v_j}\vee\overline x$, and
$\overline {v_i}\vee\overline {v_j}\vee x$ where $x$ is again a new
variable.  Note that these two transformation preserve the number of satisfying
assignments, since in the former the new variable $x$ has always the truth value
of $\overline{v_i}$ while in the latter $x$ is always false in any satisfying
assignment of this modified~instance.

In what follows, we describe a collection ${\cal Q}_I$ of quartet trees arising
from the instance $I$, and prove the following theorem.  (We present the proof
as Section~\ref{sec:proof-one}.)

\begin{theorem}\label{thm:one}
There is a one-to-one correspondence between satisfying assignments of the
instance $I$ and minimal chordal sandwiches of $({\rm int^*}({\cal Q}_I),{\rm
forb}({\cal Q}_I))$.
\end{theorem}

To simplify the presentation, we shall denote literals by capital letters $X$,
$Y$, etc., and indicate their negations by $\overline X$, $\overline Y$, etc.
(For instance, if $X=v_i$ then $\overline X=\overline{v_i}$, and if
$X=\overline{v_i}$ then $\overline X=v_i$.)

A {\em truth assignment} for the instance $I$ is a mapping
$\sigma:\{v_1,\ldots,v_n\}\rightarrow \{0,1\}$ where 0 and 1 represent {\em
false} and {\em true}, respectively.  To simplify the notation, we write $v_i=0$
and $v_i=1$ in place of $\sigma(v_i)=0$ and $\sigma(v_i)=1$, respectively, and
extend this notation to literals $X$,$Y$, etc., i.e., write $X=0$ and $X=1$ in
place of $\sigma(X)=0$ and $\sigma(X)=1$, respectively.  A truth assignment
$\sigma$ is a {\em satisfying assignment for $I$} if in each clause ${\cal C}_j$
exactly one the three literals evalues to true.  That is, for each clause ${\cal
C}_j=X\vee Y\vee Z$, either $X=1$, $Y=0$, $Z=0$, or $X=0$, $Y=1$, $Z=0$, or
$X=0$, $Y=0$, $Z=1$.

For each $i\in\{1\ldots n\}$, we let $\Delta_i$ denote all indices $j$ such that
$v_i$ or $\overline {v_i}$ appears in the clause ${\cal C}_j$.
Let ${\cal X}_I$ be the set consisting of the following~elements:
\vspace{-1ex}

\begin{enumerate}[a)]
\item $\alpha_{v_i}$, $\alpha_{\overline{v_i}}$ for each $i\in \{1\ldots n\}$,
\vspace{0.5ex}
\item $\beta^j_{v_i}$, $\beta^j_{\overline{v_i}}$ for each $i\in \{1\ldots n\}$
and each $j\in \Delta_i$,
\item $\gamma^j_1$, $\gamma^j_2$, $\gamma^j_3$, $\lambda^j$ for each
$j\in\{1\ldots m\}$,\vspace{0.5ex}
\item $\delta$ and $\mu$.
\vspace{-0.5ex}
\end{enumerate}

Consider the following collection of 2-element subsets of ${\cal X}_I$:
\vspace{-1ex}

\begin{enumerate}[a)]

\item 
$B=\Big\{\mu,\delta\Big\}$,
\vspace{0.5ex}

\item for each $i\in\{1,\ldots,n\}$:

$H_{v_i}$=$\{\alpha_{v_i}$, $\delta\Big\}$, 
$H_{\overline{v_i}}$=$\{\alpha_{\overline{v_i}}$, $\delta\Big\}$,
$A_i=\Big\{\alpha_{v_i}, \alpha_{\overline{v_i}}\Big\}$,

$S^j_{v_i}=\Big\{\alpha_{v_i},\beta^j_{v_i}\Big\}$,
$S^j_{\overline{v_i}}=\Big\{\alpha_{\overline{v_i}}, \beta^j_{\overline{v_i}}\Big\}$
for all $j\in\Delta_i$

\item for each $j\in\{1\ldots m\}$ where $C_j=X\vee Y\vee Z$:

\begin{tabular}{@{}lll}
$K^j_{\overline X}=\Big\{\beta^j_X,\gamma^j_1\Big\}$,&
$K^j_{\overline Y}=\Big\{\beta^j_Y,\gamma^j_2\Big\}$,&
$K^j_{\overline Z}=\Big\{\beta^j_Z,\gamma^j_3\Big\}$,\\

$K^j_X=\Big\{\beta^j_{\overline {X}}, \lambda^j\Big\}$,&
$K^j_Y=\Big\{\beta^j_{\overline {Y}}, \lambda^j\Big\}$,&
$K^j_Z=\Big\{\beta^j_{\overline {Z}}, \lambda^j\Big\}$,\\

$L^j_X=\Big\{\beta^j_{\overline {X}},\gamma^j_2\Big\}$,&
$L^j_Y=\Big\{\beta^j_{\overline {Y}},\gamma^j_3\Big\}$,&
$L^j_Z=\Big\{\beta^j_{\overline {Z}},\gamma^j_1\Big\}$,\\

$D^j_1=\Big\{\gamma^j_1,\lambda^j\Big\}$,&
$D^j_2=\Big\{\gamma^j_2,\lambda^j\Big\}$,&
$D^j_3=\Big\{\gamma^j_3,\lambda^j\Big\}$,\quad
$F^j=\Big\{\lambda^j,\mu\Big\}$
\end{tabular}

\end{enumerate}

The collection ${\cal Q}_I$ of quartet trees is defined as follows:
\medskip

\lefteqn{{\cal Q}_I=\bigcup_{i\in\{1\ldots n\}}\Big\{A_i|B\Big\}\cup
\hspace{-0.3em}
\bigcup_{j\in\{1\ldots m\}}\hspace{-0.3em}\Big\{D^j_1|B, D^j_2|B, D^j_3|B\Big\}}
\vspace{0.3ex}

\lefteqn{\cup\bigcup_{\substack{i\in\{1\ldots n\}\\j,j'\in\Delta_i}}\Big\{
S^j_{v_i}|S^{j'}_{\overline{v_i}}\Big\}
\cup
\hspace{-1.7em}\bigcup_{\substack{i\in\{1\ldots n\}\\j,j'\in\Delta_i{\rm ~and~}j<j'}}
\hspace{-1.7em}
\Big\{S^j_{v_i}|K^{j'}_{\overline{v_i}}, S^j_{\overline
{v_i}}|K^{j'}_{v_i}\Big\}\cup
\hspace{-1.7em}
\bigcup_{\substack{i\in\{1\ldots n\}\\ j\in\Delta_i~{\rm and}~
j<j'\leq m}}\hspace{-1.7em}\Big\{K^j_{\overline{v_i}}|F^{j'}, K^j_{v_i}|F^{j'}\Big\}}
\vspace{0.5ex}
 
\lefteqn{\cup\bigcup_{\substack{1\leq i'<i\leq n\\
j\in\Delta_i}}\Big\{
H_{v_{i'}}|S^j_{v_i}, H_{\overline{v_{i'}}}|S^j_{v_i},
H_{v_{i'}}|S^j_{\overline{v_i}},
H_{\overline{v_{i'}}}|S^j_{\overline{v_i}}\Big\}\cup
\bigcup_{\substack{i\in\{1\ldots n\}\\j\in\{1\ldots m\}}}
\Big\{H_{\overline{v_i}}|F^j, H_{v_i}|F^j\Big\}}

\lefteqn{\cup\bigcup_{\substack{j\in\{1\ldots m\}\\{\rm where~} {\cal C}_j=X\vee Y\vee Z}}\left\{
\begin{minipage}{0.67\textwidth}
\begin{tabular}{@{}l@{~}l@{~}l@{~}l@{~}l@{}l}
$K^j_{\overline X}| K^j_X$,&$K^j_{\overline Y}| K^j_Y$,&$K^j_{\overline Z}|K^j_Z$,&
$K^j_{\overline X}| L^j_X$,&$K^j_{\overline Y}| L^j_Y$,&$K^j_{\overline
Z}|L^j_Z$\vspace{0.5ex}\\
$S^j_Y|K^j_X$,&$S^j_Z|K^j_Y$,&$S^j_X|K^j_Z$,&
$S^j_Z|L^j_X$,&$S^j_X|L^j_Y$,&$S^j_Y|L^j_Z$\end{tabular}
\end{minipage}\right\}}
\bigskip

\begin{figure}[t!]
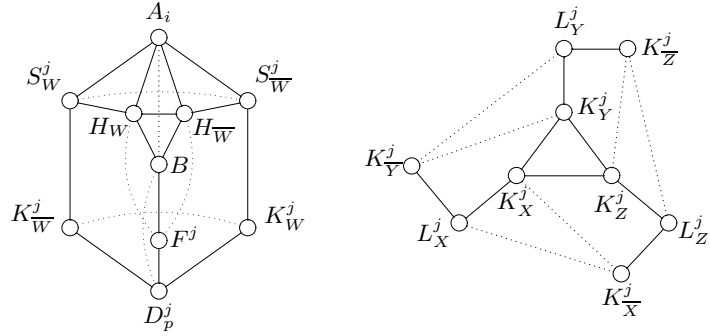

\centering
\small 
\raisebox{-2ex}{$\xy/r4pc/:
(-0.7,-0.5)*[o][F]{\phantom{S}}="k'";
(0.7,-0.5)*[o][F]{\phantom{S}}="k";
(-0.7,0.5)*[o][F]{\phantom{S}}="s";
(0.7,0.5)*[o][F]{\phantom{S}}="s'";
(0,-1)*[o][F]{\phantom{S}}="d";
(0,1)*[o][F]{\phantom{S}}="a";
(0,0)*[o][F]{\phantom{S}}="b";
(-0.2,0.4)*[o][F]{\phantom{S}}="h";
(0.2,0.4)*[o][F]{\phantom{S}}="h'";
(0,-0.6)*[o][F]{\phantom{S}}="f";
{\ar@{-} "k";"d"};
{\ar@{-} "k'";"d"};
{\ar@{-} "s";"a"};
{\ar@{-} "s'";"a"};
{\ar@{-} "s";"k'"};
{\ar@{-} "s'";"k"};
{\ar@{-} "s";"h"};
{\ar@{-} "h";"h'"};
{\ar@{-} "h'";"s'"};
{\ar@{-} "h'";"a"};
{\ar@{-} "h";"a"};
{\ar@{-} "h";"b"};
{\ar@{-} "h'";"b"};
{\ar@{-} "f";"b"};
{\ar@{-} "f";"d"};
{\ar@{.} "a";"b"};
{\ar@{.}@/_0.5pc/ "b";"d"};
{\ar@{.}@/^0.3pc/ "s";"s'"};
{\ar@{.}@/_0.5pc/ "k";"k'"};
{\ar@{.}@/_0.6pc/ "h";"f"};
{\ar@{.}@/^0.6pc/ "h'";"f"};
"a"+(0,0.2)*{A_i};
"d"+(0,-0.2)*{D^j_p};
"s"+(-0.2,0.2)*{S^j_W};
"s'"+(0.2,0.2)*{S^j_{\overline W}};
"h"+(-0.19,-0.1)*{H_W};
"h'"+(0.23,-0.13)*{H_{\overline W}};
"f"+(0.2,0.03)*{F^j};
"b"+(0.16,0.0)*{B};
"k"+(0.3,0.1)*{K^j_W};
"k'"+(-0.3,0.1)*{K^j_{\overline W}};
\endxy$}
\qquad
$\xy/r2pc/:
(-0.75,-0.5)*[o][F]{\phantom{S}}="x";
(-1.65,-1.25)*[o][F]{\phantom{S}}="lx";
(0,0.5)*[o][F]{\phantom{S}}="y";
(0,1.5)*[o][F]{\phantom{S}}="ly";
(0.75,-0.5)*[o][F]{\phantom{S}}="z";
(1.65,-1.25)*[o][F]{\phantom{S}}="lz";
(0.9,-2.05)*[o][F]{\phantom{S}}="x'";
(-2.4,-0.35)*[o][F]{\phantom{S}}="y'";
(1,1.5)*[o][F]{\phantom{S}}="z'";
{\ar@{-} "x";"lx"};
{\ar@{-} "lx";"y'"};
{\ar@{-} "y";"ly"};
{\ar@{-} "ly";"z'"};
{\ar@{-} "z";"lz"};
{\ar@{-} "lz";"x'"};
{\ar@{-} "x";"y"};
{\ar@{-} "y";"z"};
{\ar@{-} "z";"x"};
{\ar@{.} "x'";"x"};
{\ar@{.} "x'";"lx"};
{\ar@{.} "y'";"y"};
{\ar@{.} "y'";"ly"};
{\ar@{.} "z'";"z"};
{\ar@{.} "z'";"lz"};
"x'"+(0,-0.4)*{K^j_{\overline X}};
"y'"+(-0.45,0.1)*{K^j_{\overline Y}};
"z'"+(0.5,0)*{K^j_{\overline Z}};
"lx"+(-0.4,-0.15)*{L^j_X};
"x"+(0,-0.4)*{K^j_X};
"ly"+(0.1,0.45)*{L^j_Y};
"y"+(0.5,0.1)*{K^j_Y};
"lz"+(0.4,-0.1)*{L^j_Z};
"z"+(0,-0.45)*{K^j_Z};
\endxy$
\caption{Two configurations from of the graph ${\rm int^*}({\cal Q}_I)$.\label{fig:1}}
\end{figure}

Note that in each clause ${\cal C}_j=X\vee Y\vee Z$ there is a particular type
of symmetry between the literals $X$, $Y$, and $Z$.  In particular, if we
replace, in the above, the incices $X$, $Y$, $Z$ and 1, 2, 3 as follows: $X
\rightarrow Y \rightarrow Z\rightarrow X$ and $1\rightarrow 2 \rightarrow
3\rightarrow 1$, we obtain precisely the same definition of ${\cal Q}_I$ as the
above.  We shall refer to this as the {\em rotational symmetry} between $X$,
$Y$, $Z$.

\section{Unique trees}\label{sec:unique-trees}
Let $T_I$ be the tree defined as follows: (for illustration, see Figures \ref{fig:3b} and
\ref{fig:3a})\vspace{1ex}

\noindent $V(T_I)=\Big\{y_0,y_1,y'_1,\ldots,y_n,y'_n\Big\}
\cup\Big\{a_1,a'_1,\ldots,a_n,a'_n\Big\}
\cup\Big\{u_0,u_1,\ldots,u_m\Big\}\\
\mbox{}\hspace{4em}\cup\Big\{x^j_1,x^j_2,x^j_3,x^j_4,x^j_5,x^j_6,
b^j_1, b^j_2, b^j_3, g^j_1, g^j_2, g^j_3, \ell^j
\Big\}_{j=1}^m\cup
\Big\{c^j_i, z^j_i~|~j\in \Delta_i\Big\}_{i=1}^n$

\noindent $E(T_I)=\Big\{y_1y'_1, y_2y'_2,\ldots,y_ny'_n\Big\} \cup$
$\Big\{a_1y'_1, a_2y'_2,\ldots a_ny'_n\Big\}$
$\cup\Big\{c^j_i z^j_i~|~j\in\Delta_i\Big\}_{i=1}^n$\\
$\cup\Big\{y_0y_1, y_1y_2, y_2y_3, \ldots, y_{n-1}y_n\Big\}\cup\Big\{y_nu_1,
u_1u_2, u_2u_3, \ldots, u_{m-1}u_m, u_mu_0\Big\}$\\
$\cup\Big\{u_jx^j_1,x^j_1x^j_2,x^j_2x^j_3,x^j_2x^j_4,x^j_4x^j_5,x^j_4x^j_6,
b^j_1 x^j_6, b^j_2 x^j_3, b^j_3 x^j_5, g^j_1 x^j_6, g^j_2 x^j_1, g^j_3 x^j_3,
\ell^j x^j_5 \Big\}_{j=1}^m$\\
$\cup\Big\{a'_iz_i^{j_1}, z_i^{j_1}z_i^{j_2}, \ldots, z_i^{j_{t-1}}z_i^{j_t},
z_i^{j_t}y'_i~|~\mbox{where $j_1<j_2<\ldots<j_t$ are elements of
$\Delta_i$}\Big\}_{i=1}^n$

\vspace{2ex}

Let $\sigma$ be a satisfying assignment for the instance $I$, and let
$\phi_\sigma$ be the mapping of ${\cal X}_I$ to $V(T_I)$ defined as
follows:

\begin{figure}[b!]
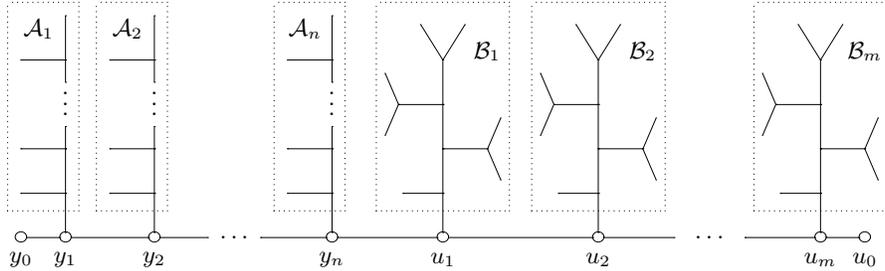

\mbox{}\hfill
$\xy/r1.4pc/:
(0,0)*[o][F]{\phantom{s}}="delta";
(1,0)*[o][F]{\phantom{s}}="y1";
(3,0)*[o][F]{\phantom{s}}="y2";
(7,0)*[o][F]{\phantom{s}}="yn";
(9.5,0)*[o][F]{\phantom{s}}="u1";
(13,0)*[o][F]{\phantom{s}}="u2";
(18,0)*[o][F]{\phantom{s}}="um";
(19,0)*[o][F]{\phantom{s}}="mu";
(4.8,0)*{\ldots};
(15.5,0)*{\ldots};
"u1"+(0,-0.5)*{u_1};
"u2"+(0,-0.5)*{u_2};
"um"+(0,-0.5)*{u_m};
"y1"+(0,-0.5)*{y_1};
"y2"+(0,-0.5)*{y_2};
"yn"+(0,-0.5)*{y_n};
"delta"+(0,-0.5)*{y_0};
"mu"+(0,-0.5)*{u_0};
{\ar@{-} "delta";"y1"};
{\ar@{-} "y1";"y2"};
{\ar@{-} "y2";"y2"+(1.2,0)};
{\ar@{-} "yn";"yn"-(1.6,0)};
{\ar@{-} "yn";"u1"};
{\ar@{-} "u1";"u2"};
{\ar@{-} "u2";"u2"+(1.7,0)};
{\ar@{-} "um";"um"-(1.7,0)};
{\ar@{-} "um";"mu"};
{\ar@{-} "y1";"y1"+(0,1)};
{\ar@{-} "y1"+(0,1);"y1"+(-1,1)};
{\ar@{-} "y1"+(0,2);"y1"+(-1,2)};
{\ar@{-} "y1"+(0,4);"y1"+(-1,4)};
{\ar@{-} "y1"+(0,4);"y1"+(0,5)};
{\ar@{-} "y1"+(0,1);"y1"+(0,2)};
{\ar@{-} "y1"+(0,2);"y1"+(0,2.5)};
{\ar@{-} "y1"+(0,3.5);"y1"+(0,4)};
"y1"+(0,3.2)*{\vdots};
{\ar@{-} "y2";"y2"+(0,1)};
{\ar@{-} "y2"+(0,1);"y2"+(-1,1)};
{\ar@{-} "y2"+(0,2);"y2"+(-1,2)};
{\ar@{-} "y2"+(0,4);"y2"+(-1,4)};
{\ar@{-} "y2"+(0,4);"y2"+(0,5)};
{\ar@{-} "y2"+(0,1);"y2"+(0,2)};
{\ar@{-} "y2"+(0,2);"y2"+(0,2.5)};
{\ar@{-} "y2"+(0,3.5);"y2"+(0,4)};
"y2"+(0,3.2)*{\vdots};
{\ar@{-} "yn";"yn"+(0,1)};
{\ar@{-} "yn"+(0,1);"yn"+(-1,1)};
{\ar@{-} "yn"+(0,2);"yn"+(-1,2)};
{\ar@{-} "yn"+(0,4);"yn"+(-1,4)};
{\ar@{-} "yn"+(0,4);"yn"+(0,5)};
{\ar@{-} "yn"+(0,1);"yn"+(0,2)};
{\ar@{-} "yn"+(0,2);"yn"+(0,2.5)};
{\ar@{-} "yn"+(0,3.5);"yn"+(0,4)};
"yn"+(0,3.2)*{\vdots};
{\ar@{-} "u1";"u1"+(0,1)};
{\ar@{-} "u1"+(0,1);"u1"+(0,2)};
{\ar@{-} "u1"+(0,1);"u1"+(-0.9,1)};
{\ar@{-} "u1"+(0,2);"u1"+(1,2)};
{\ar@{-} "u1"+(1,2);"u1"+(1.3,1.3)};
{\ar@{-} "u1"+(1,2);"u1"+(1.3,2.7)};
{\ar@{-} "u1"+(0,2);"u1"+(0,3)};
{\ar@{-} "u1"+(0,3);"u1"+(-1,3)};
{\ar@{-} "u1"+(-1,3);"u1"+(-1.3,3.7)};
{\ar@{-} "u1"+(-1,3);"u1"+(-1.3,2.3)};
{\ar@{-} "u1"+(0,3);"u1"+(0,4)};
{\ar@{-} "u1"+(-0.5,4.8);"u1"+(0,4)};
{\ar@{-} "u1"+(0.5,4.8);"u1"+(0,4)};
{\ar@{-} "u2";"u2"+(0,1)};
{\ar@{-} "u2"+(0,1);"u2"+(0,2)};
{\ar@{-} "u2"+(0,1);"u2"+(-0.9,1)};
{\ar@{-} "u2"+(0,2);"u2"+(1,2)};
{\ar@{-} "u2"+(1,2);"u2"+(1.3,1.3)};
{\ar@{-} "u2"+(1,2);"u2"+(1.3,2.7)};
{\ar@{-} "u2"+(0,2);"u2"+(0,3)};
{\ar@{-} "u2"+(0,3);"u2"+(-1,3)};
{\ar@{-} "u2"+(-1,3);"u2"+(-1.3,3.7)};
{\ar@{-} "u2"+(-1,3);"u2"+(-1.3,2.3)};
{\ar@{-} "u2"+(0,3);"u2"+(0,4)};
{\ar@{-} "u2"+(-0.5,4.8);"u2"+(0,4)};
{\ar@{-} "u2"+(0.5,4.8);"u2"+(0,4)};
{\ar@{-} "um";"um"+(0,1)};
{\ar@{-} "um"+(0,1);"um"+(0,2)};
{\ar@{-} "um"+(0,1);"um"+(-0.9,1)};
{\ar@{-} "um"+(0,2);"um"+(1,2)};
{\ar@{-} "um"+(1,2);"um"+(1.3,1.3)};
{\ar@{-} "um"+(1,2);"um"+(1.3,2.7)};
{\ar@{-} "um"+(0,2);"um"+(0,3)};
{\ar@{-} "um"+(0,3);"um"+(-1,3)};
{\ar@{-} "um"+(-1,3);"um"+(-1.3,3.7)};
{\ar@{-} "um"+(-1,3);"um"+(-1.3,2.3)};
{\ar@{-} "um"+(0,3);"um"+(0,4)};
{\ar@{-} "um"+(-0.5,4.8);"um"+(0,4)};
{\ar@{-} "um"+(0.5,4.8);"um"+(0,4)};
{\ar@{.} "y1"+(0.3,0.6);"y1"+(-1.3,0.6)};
{\ar@{.} "y1"+(-1.3,5.3);"y1"+(-1.3,0.6)};
{\ar@{.} "y1"+(-1.3,5.3);"y1"+(0.3,5.3)};
{\ar@{.} "y1"+(0.3,5.3);"y1"+(0.3,0.6)};
"y1"+(-0.6,4.7)*{{\cal A}_1};
{\ar@{.} "y2"+(0.3,0.6);"y2"+(-1.3,0.6)};
{\ar@{.} "y2"+(-1.3,5.3);"y2"+(-1.3,0.6)};
{\ar@{.} "y2"+(-1.3,5.3);"y2"+(0.3,5.3)};
{\ar@{.} "y2"+(0.3,5.3);"y2"+(0.3,0.6)};
"y2"+(-0.6,4.7)*{{\cal A}_2};
{\ar@{.} "yn"+(0.3,0.6);"yn"+(-1.3,0.6)};
{\ar@{.} "yn"+(-1.3,5.3);"yn"+(-1.3,0.6)};
{\ar@{.} "yn"+(-1.3,5.3);"yn"+(0.3,5.3)};
{\ar@{.} "yn"+(0.3,5.3);"yn"+(0.3,0.6)};
"yn"+(-0.6,4.7)*{{\cal A}_n};
{\ar@{.} "u1"+(1.5,0.6);"u1"+(-1.5,0.6)};
{\ar@{.} "u1"+(-1.5,5.3);"u1"+(-1.5,0.6)};
{\ar@{.} "u1"+(-1.5,5.3);"u1"+(1.5,5.3)};
{\ar@{.} "u1"+(1.5,5.3);"u1"+(1.5,0.6)};
"u1"+(1.0,4.2)*{{\cal B}_1};
{\ar@{.} "u2"+(1.5,0.6);"u2"+(-1.5,0.6)};
{\ar@{.} "u2"+(-1.5,5.3);"u2"+(-1.5,0.6)};
{\ar@{.} "u2"+(-1.5,5.3);"u2"+(1.5,5.3)};
{\ar@{.} "u2"+(1.5,5.3);"u2"+(1.5,0.6)};
"u2"+(1.0,4.2)*{{\cal B}_2};
{\ar@{.} "um"+(1.5,0.6);"um"+(-1.5,0.6)};
{\ar@{.} "um"+(-1.5,5.3);"um"+(-1.5,0.6)};
{\ar@{.} "um"+(-1.5,5.3);"um"+(1.5,5.3)};
{\ar@{.} "um"+(1.5,5.3);"um"+(1.5,0.6)};
"um"+(1.0,4.2)*{{\cal B}_m};
\endxy$
\hfill\mbox{}
\caption{The tree $T_I$.\label{fig:3b}}
\end{figure}

\begin{enumerate}[a)]
\item for each $i\in\{1\ldots n\}$:

if $v_i=1$, then $\phi_\sigma(\alpha_{v_i})=a_i$,
$\phi_\sigma(\alpha_{\overline{v_i}})=a'_i$, and
$\phi_\sigma(\beta^j_{\overline{v_i}})=c^j_i$ for all $j\in\Delta_i$,

if $v_i=0$, then $\phi_\sigma(\alpha_{\overline{v_i}})=a_i$,
$\phi_\sigma(\alpha_{v_i})=a'_i$, and
$\phi_\sigma(\beta^j_{v_i})=c^j_i$ for all $j\in\Delta_i$,

\item for each $j\in\{1\ldots m\}$ where ${\cal C}_j=X\vee Y\vee Z$:

if $X=1$, then 
$\phi_\sigma(\beta^j_X)=b^j_1$,
$\phi_\sigma(\beta^j_{\overline Y})=b^j_2$,
$\phi_\sigma(\beta^j_{\overline Z})=b^j_3$,\\
\mbox{}\hspace{6.5em}$\phi_\sigma(\gamma^j_1)=g^j_1$,~
$\phi_\sigma(\gamma^j_2)=g^j_2$,~
$\phi_\sigma(\gamma^j_3)=g^j_3$,~~
$\phi_\sigma(\lambda^j)=\ell_j$,

if $Y=1$, then 
$\phi_\sigma(\beta^j_Y)=b^j_1$,
$\phi_\sigma(\beta^j_{\overline Z})=b^j_2$,
$\phi_\sigma(\beta^j_{\overline X})=b^j_3$,\\
\mbox{}\hspace{6.5em}$\phi_\sigma(\gamma^j_2)=g^j_1$,~
$\phi_\sigma(\gamma^j_3)=g^j_2$,~
$\phi_\sigma(\gamma^j_1)=g^j_3$,~~
$\phi_\sigma(\lambda^j)=\ell_j$,

if $Z=1$, then 
$\phi_\sigma(\beta^j_Z)=b^j_1$,
$\phi_\sigma(\beta^j_{\overline X})=b^j_2$,
$\phi_\sigma(\beta^j_{\overline Y})=b^j_3$,\\
\mbox{}\hspace{6.5em}$\phi_\sigma(\gamma^j_3)=g^j_1$,~
$\phi_\sigma(\gamma^j_1)=g^j_2$,~
$\phi_\sigma(\gamma^j_2)=g^j_3$,~~
$\phi_\sigma(\lambda^j)=\ell_j$,
\item $\phi_\sigma(\delta)=y_0$ and $\phi_\sigma(\mu)=u_0$.
\end{enumerate}

\begin{figure}[t!]
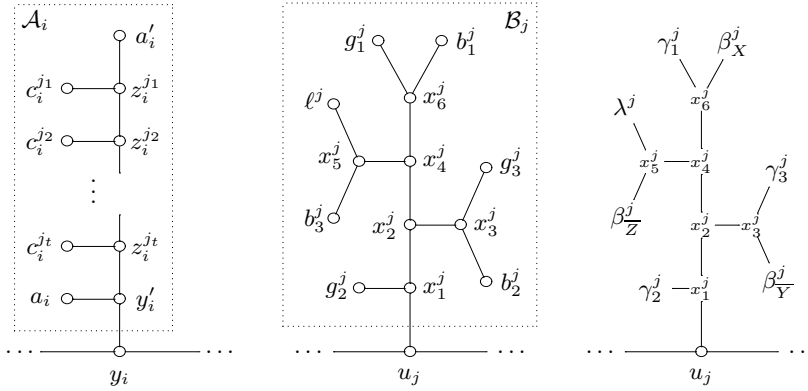

\mbox{}\hfill
$\xy/r1.66pc/:
(1,0)*[o][F]{\phantom{s}}="y1";
(1,1)*[o][F]{\phantom{s}}="y'1";
(0,1)*[o][F]{\phantom{s}}="a1";
(1,6)*[o][F]{\phantom{s}}="a'1";
(1,5)*[o][F]{\phantom{s}}="z11";
(1,4)*[o][F]{\phantom{s}}="z12";
(1,2)*[o][F]{\phantom{s}}="z1t";
(0,5)*[o][F]{\phantom{s}}="c11";
(0.5,3.2)*{\vdots};
(0,4)*[o][F]{\phantom{s}}="c12";
(0,2)*[o][F]{\phantom{s}}="c1t";
(-0.9,0)*{\ldots};
(2.9,0)*{\ldots};
{\ar@{-} "y1";"y'1"};
{\ar@{-} "y'1";"a1"};
{\ar@{-} "y'1";"z1t"};
{\ar@{-} "z12";"z11"};
{\ar@{-} "z12";"z12"-(0,0.6)};
{\ar@{-} "z1t";"z1t"+(0,0.6)};
{\ar@{-} "z11";"a'1"};
{\ar@{-} "z11";"c11"};
{\ar@{-} "z12";"c12"};
{\ar@{-} "z1t";"c1t"};
{\ar@{-} "y1";"y1"+(1.5,0)};
{\ar@{-} "y1";"y1"-(1.5,0)};
"a1"+(-0.5,0)*{a_i};
"c11"+(-0.5,0)*{c^{j_1}_i};
"c12"+(-0.5,0)*{c^{j_2}_i};
"c1t"+(-0.5,0)*{c^{j_t}_i};
"z11"+(0.5,0)*{z^{j_1}_i};
"z12"+(0.5,0)*{z^{j_2}_i};
"z1t"+(0.5,0)*{z^{j_t}_i};
"y'1"+(0.5,0)*{y'_i};
"y1"+(0,-0.5)*{y_i};
"a'1"+(0.5,0)*{a'_i};
{\ar@{.} "y1"+(1,0.4);"y1"+(-2,0.4)};
{\ar@{.} "y1"+(-2,6.6);"y1"+(-2,0.4)};
{\ar@{.} "y1"+(-2,6.6);"y1"+(1,6.6)};
{\ar@{.} "y1"+(1,6.6);"y1"+(1,0.4)};
"y1"+(-1.6,6.3)*{{\cal A}_i};
\endxy
\qquad
\xy/r2pc/:
(0.8,0)*[o][F]{\phantom{s}}="u1";
(0.8,1)*[o][F]{\phantom{s}}="x11";
(0.8,2)*[o][F]{\phantom{s}}="x12";
(1.6,2)*[o][F]{\phantom{s}}="x13";
(0.8,3)*[o][F]{\phantom{s}}="x14";
(0,3)*[o][F]{\phantom{s}}="x15";
(0.8,4)*[o][F]{\phantom{s}}="x16";
(0,1)*[o][F]{\phantom{s}}="g12";
(-0.4,2.1)*[o][F]{\phantom{s}}="b13";
(2.0,1.1)*[o][F]{\phantom{s}}="b12";
(2.0,2.9)*[o][F]{\phantom{s}}="g13";
(1.3,4.9)*[o][F]{\phantom{s}}="b11";
(-0.4,3.9)*[o][F]{\phantom{s}}="l1";
(0.3,4.9)*[o][F]{\phantom{s}}="g11";
{\ar@{-} "u1";"x11"};
{\ar@{-} "x11";"x12"};
{\ar@{-} "x12";"x13"};
{\ar@{-} "x12";"x14"};
{\ar@{-} "x15";"x14"};
{\ar@{-} "x16";"x14"};
{\ar@{-} "x11";"g12"};
{\ar@{-} "x13";"g13"};
{\ar@{-} "x13";"b12"};
{\ar@{-} "x15";"b13"};
{\ar@{-} "x15";"l1"};
{\ar@{-} "x16";"b11"};
{\ar@{-} "x16";"g11"};
{\ar@{-} "u1";"u1"+(1.2,0)};
{\ar@{-} "u1";"u1"+(-1.2,0)};
(-0.8,0)*{\ldots};
(2.4,0)*{\ldots};
"u1"+(0,-0.4)*{u_j};
"x11"+(0.4,0.05)*{x^j_1};
"x12"+(-0.4,0)*{x^j_2};
"x13"+(0.4,0.05)*{x^j_3};
"x14"+(0.4,0.05)*{x^j_4};
"x15"+(-0.45,0.1)*{x^j_5};
"x16"+(0.4,0.05)*{x^j_6};
"g11"+(-0.35,0.05)*{g^j_1};
"g12"+(-0.35,0.05)*{g^j_2};
"g13"+(0.4,0.1)*{g^j_3};
"b11"+(0.4,0)*{b^j_1};
"b12"+(0.4,0)*{b^j_2};
"b13"+(-0.3,0)*{b^j_3};
"l1"+(-0.3,0.1)*{\ell^j};
{\ar@{.} "u1"+(2,0.4);"u1"+(-2,0.4)};
{\ar@{.} "u1"+(-2,5.5);"u1"+(-2,0.4)};
{\ar@{.} "u1"+(-2,5.5);"u1"+(2,5.5)};
{\ar@{.} "u1"+(2,5.5);"u1"+(2,0.4)};
"u1"+(1.7,5.2)*{{\cal B}_j};
\endxy
\qquad
\xy/r2pc/:
(0.8,0)*[o][F]{\phantom{s}}="u1";
(0.8,1)*+[o]{\phantom{s}}="x11";
(0.8,2)*+[o]{\phantom{s}}="x12";
(1.6,2)*+[o]{\phantom{s}}="x13";
(0.8,3)*+[o]{\phantom{s}}="x14";
(0,3)*+[o]{\phantom{s}}="x15";
(0.8,4)*+[o]{\phantom{s}}="x16";
(0,1)*++[o]{\phantom{s}}="g12";
(-0.4,2.1)*++[o]{\phantom{s}}="b13";
(2.0,1.1)*++[o]{\phantom{s}}="b12";
(2.0,2.9)*++[o]{\phantom{s}}="g13";
(1.3,4.9)*++[o]{\phantom{s}}="b11";
(-0.4,3.9)*++[o]{\phantom{s}}="l1";
(0.3,4.9)*++[o]{\phantom{s}}="g11";
{\ar@{-} "u1";"x11"};
{\ar@{-} "x11";"x12"};
{\ar@{-} "x12";"x13"};
{\ar@{-} "x12";"x14"};
{\ar@{-} "x15";"x14"};
{\ar@{-} "x16";"x14"};
{\ar@{-} "x11";"g12"};
{\ar@{-} "x13";"g13"};
{\ar@{-} "x13";"b12"};
{\ar@{-} "x15";"b13"};
{\ar@{-} "x15";"l1"};
{\ar@{-} "x16";"b11"};
{\ar@{-} "x16";"g11"};
{\ar@{-} "u1";"u1"+(1.2,0)};
{\ar@{-} "u1";"u1"+(-1.2,0)};
(-0.8,0)*{\ldots};
(2.4,0)*{\ldots};
"u1"+(0,-0.4)*{u_j};
"x11"*{_{x^j_1}};
"x12"*{_{x^j_2}};
"x13"*{_{x^j_3}};
"x14"*{_{x^j_4}};
"x15"*{_{x^j_5}};
"x16"*{_{x^j_6}};
"g11"*{\gamma^j_1};
"g12"*{\gamma^j_2};
"g13"*{\gamma^j_3};
"b11"*{\beta^j_X};
"b12"*{\beta^j_{\overline Y}};
"b13"*{\beta^j_{\overline Z}};
"l1"*{\lambda^j};
\endxy$
\hfill\mbox{}
\caption{{\em a)} the subtree ${\cal A}_i$ for the variable
$v_i$, {\em b)} the subtree ${\cal B}_j$ for the clause ${\cal C}_j$, {\em c)}
the subtree for ${\cal C}_j=X\vee Y\vee Z$ and assignment
$\sigma(X)=1$, $\sigma(Y)=\sigma(Z)=0$\label{fig:3a}}
\end{figure}

\begin{theorem}\label{thm:unique-trees}
If $\sigma$ is a satisfying assignment for $I$, then
${\cal T}_\sigma=(T_I,\phi_\sigma)$ is a ternary phylogenetic ${\cal X}_I$-tree
that displays ${\cal Q}_I$ and is distinguished by ${\cal Q}_I$.
\end{theorem}
\begin{proof}
Let $\sigma$ be a satisfying assignment for $I$, i.e., for each clause ${\cal
C}_j=X\vee Y\vee Z$, either $X=1$, $Y=Z=0$, or $Y=1$, $X=Z=0$, or $Z=1$,
$X=Y=0$. For each $i\in\{1\ldots n\}$, let ${\cal A}_i=\{a_i$,\,$a'_i$,\,$y'_i$,
$z^{j_1}_i$,\,\ldots, $z^{j_t}_i$,$c^{j_1}_i$,\,\ldots, $c^{j_t}_i\}$ where
$\Delta_i=\{j_1,\ldots,j_t\}$, and for each $j\in\{1\ldots m\}$, let ${\cal
B}_j=\{x^j_1$,\,$x^j_2$,\,$x^j_3$,\,$x^j_4$,\,$x^j_5$,\,$x^j_6$,
$g^j_1$,\,$g^j_2$,\,$g^j_3$,\,$b^j_1$,\,$b^j_2$,\,$b^j_3$,\,$\ell^j\}$.

\newpage
\mbox{}\vskip -6ex

It is not difficult to see that $\phi_\sigma$ defines a bijection between the
elements of ${\cal X}_I$ and the leaves of $T_I$. For instance, for each
$i\in\{1\ldots n\}$, we note that $\{\phi(\alpha_{v_i})$,
$\phi(\alpha_{\overline{v_i}})\}=\{a_i$, $a'_i\}$, and for each $j\in\Delta_i$,
either $\phi_\sigma(\beta^j_{v_i})=c^j_i$ and
$\phi_\sigma(\beta^j_{\overline{v_i}})\in\{b^j_1,b^j_2,b^j_3\}$, or
$\phi_\sigma(\beta^j_{\overline{v_i}})=c^j_i$ and
$\phi_\sigma(\beta^j_{v_i})\in\{b^j_1,b^j_2,b^j_3\}$.  Also, for each
$j\in\{1\ldots m\}$, we have $\phi_\sigma(\lambda^j)=\ell^j$, and
$\{\phi_\sigma(\gamma^j_1), \phi_\sigma(\gamma^j_2),
\phi_\sigma(\gamma^j_3)\}=\{g^j_1,g^j_2,g^j_3\}$.
Further, it can be readily verified that $T_I$ is a ternary tree.  Thus, ${\cal
T}_\sigma=(T_I,\phi_\sigma)$ is indeed a ternary phylogenetic ${\cal X}_I$-tree.
First, we show that it displays ${\cal Q}_I$.

Consider $A_i|B$ for $i\in\{1\ldots n\}$. Recall that
$A_i=\{\alpha_{v_i},\alpha_{\overline{v_i}}\}$, $B=\{\delta,\mu\}$, and that
$\{\phi_\sigma(\alpha_{v_i}),\phi_\sigma(\alpha_{\overline{v_i}})\}=\{a_i,a'_i\}$.
Also, $\phi_\sigma(\delta)=y_0$ and $\phi_\sigma(\mu)=u_0$.  Observe that
$a_i,a'_i\in{\cal A}_i$. Hence, both $a_i, a_i'$ are in one connected component
of $T_I-y_iy'_i$ whereas $y_0,u_0$ are in another component. Thus, ${\cal
T}_\sigma$ indeed displays $A_i|B$.

Next, consider $D^j_p|B$ for $j\in\{1\ldots m\}$ and $p\in\{1\ldots 3\}$.
Recall that $D^j_p=\{\gamma^j_p,\lambda^j\}$, and
$\phi_\sigma(\gamma^j_p)\in{\cal B}_j$, $\phi_\sigma(\lambda^j)\in{\cal B}_j$.
Also, $B=\{\delta,\mu\}$ and $\phi_\sigma(\delta)=y_0$, $\phi_\sigma(\mu)=u_0$.
Thus both $\phi_\sigma(\gamma^j_p)$, $\phi_\sigma(\lambda^j)$ are in one
component of $T_I-u_jx^j_1$ whereas $y_0,u_0$ are in another component. This shows
that ${\cal T}_\sigma$ displays $D^j_p|B$.

Now, we look at $S^j_{v_i}|S^{j'}_{\overline{v_i}}$ where $i\in\{1\ldots n\}$
and $j,j'\in\Delta_i$. Recall that $S^j_{v_i}=\{\alpha_{v_i},\beta^j_{v_i}\}$
and $S^{j'}_{\overline{v_i}}=\{\alpha_{\overline{v_i}},
\beta^{j'}_{\overline{v_i}}\}$.  By symmetry, we may assume that $v_i=1$. Then
$\phi_\sigma(\alpha_{v_i})=a_i$, $\phi_\sigma(\alpha_{\overline{v_i}})=a'_i$,
$\phi_\sigma(\beta^j_{v_i})\in{\cal B}_j$, and
$\phi_\sigma(\beta^{j'}_{\overline{v_i}})=c^{j'}_i$.  Let $j_t$ denote the
largest element in $\Delta_i$. Then, both $a'_i$,$c^{j'}_i$ are in one component
of $T_I-y'_iz^{j_t}_i$ whereas $a_i$ and $\phi_\sigma(\beta^j_{v_i})$ are in a
different component. Thus, ${\cal T}_\sigma$ displays
$S^j_{v_i}|S^{j'}_{\overline{v_i}}$.

Next, consider $S^j_{v_i}|K^{j'}_{\overline{v_i}}$ and
$S^j_{\overline{v_i}}|K^{j'}_{v_i}$ for $i\in\{1\ldots n\}$ and
$j,j'\in\Delta_i$ where $j<j'$.  Recall that $K^{j'}_{\overline{v_i}}\subseteq
\{\beta^{j'}_{v_i},\gamma^{j'}_1,\gamma^{j'}_2,\gamma^{j'}_3,\lambda^{j'}\}$,
$K^{j'}_{v_i}\subseteq \{\beta^{j'}_{\overline{v_i}}, \gamma^{j'}_1,
\gamma^{j'}_2, \gamma^{j'}_3, \lambda^{j'}\}$,
$S^j_{v_i}=\{\alpha_{v_i},\beta^j_{v_i}\}$  and
$S^j_{\overline{v_i}}=\{\alpha_{\overline{v_i}},\beta^j_{v_i}\}$.  Again, by
symmetry, we assume $v_i=1$. So, $\phi_\sigma(\alpha_{v_i})=a_i$,
$\phi_\sigma(\alpha_{\overline{v_i}})=a'_i$,
$\phi_\sigma(\beta^j_{\overline{v_i}})=c^j_i$,
$\phi_\sigma(\beta^{j'}_{\overline{v_i}})=c^{j'}_i$,
$\phi_\sigma(\beta^j_{v_i})\in{\cal B}_j$, and
$\{\phi_\sigma(\beta^{j'}_{v_i}),\phi_\sigma(\gamma^{j'}_1),
\phi_\sigma(\gamma^{j'}_2), \phi_\sigma(\gamma^{j'}_3),
\phi_\sigma(\lambda^{j'})\}\subseteq {\cal B}_{j'}$.  Let $j_1<j_2<\ldots<j_t$
be the elements of $\Delta_i$. Since $j\in\Delta_i$, let $k$ be such that
$j=j_k$. We conclude $k<t$, since $j<j'$ and $j'\in\Delta_i$.  Thus, the
elements of $\phi_\sigma(S^j_{\overline{v_i}})$ and $\phi_\sigma(K^{j'}_{v_i})$,
respectively are in different components of $T_I-z^{j_k}_iz^{j_{k+1}}_i$. Further,
observe that~$\phi_\sigma(K^{j'}_{\overline{v_i}})\subseteq{\cal B}_{j'}$, and
since $j\neq j'$, the elements of $\phi_\sigma(S^j_{v_i})$ and
$\phi_\sigma(K^{j'}_{\overline{v_i}})$ are in different components of
$T_I-u_{j'}x^{j'}_1$. This proves that ${\cal T}_\sigma$ displays both
$S^j_{v_i}|K^{j'}_{\overline{v_i}}$ and $S^j_{\overline{v_i}}|K^{j'}_{v_i}$.

Now, consider $K^j_{\overline{v_i}}|F^{j'}$ and $K^j_{v_i}|F^{j'}$ for
$i\in\{1\ldots n\}$ and $j<j'$ where $j\in\Delta_i$.  Again, recall that
$K^j_{\overline{v_i}}\subseteq
\{\beta^j_{v_i},\gamma^j_1,\gamma^j_2,\gamma^j_3,\lambda^j\}$,
$K^j_{v_i}\subseteq \{\beta^j_{\overline{v_i}}, \gamma^j_1, \gamma^j_2,
\gamma^j_3, \lambda^j\}$, and that $F^{j'}=\{\lambda^{j'},\mu\}$.  So,
$\phi_\sigma(K^j_{\overline{v_i}})\cup \phi_\sigma(K^j_{v_i})\subseteq {\cal
A}_i\cup {\cal B}_j$ whereas $\phi_\sigma(F^{j'})\subseteq {\cal
B}_{j'}\cup\{u_0\}$.  Since $j<j'\leq m$, we conclude that
$\phi_\sigma(K^j_{\overline{v_i}})\cup \phi_\sigma(K^j_{v_i})$ and
$\phi_\sigma(F^{j'})$ are in different components of $T_I-u_ju_{j+1}$.  Thus
${\cal T}_\sigma$ displays both $K^j_{\overline{v_i}}|F^{j'}$ and
$K^j_{v_i}|F^{j'}$.

Next, we consider $H_{v_{i'}}|S^j_{v_i}$, $H_{\overline{v_{i'}}}|S^j_{v_i}$,
$H_{v_{i'}}|S^j_{\overline{v_i}}$, and
$H_{\overline{v_{i'}}}|S^j_{\overline{v_i}}$ for $1\leq i'<i\leq n$ and
$j\in\Delta_i$.  Recall that $H_{v_{i'}}=\{\alpha_{v_{i'}},\delta\}$,
$H_{\overline{v_{i'}}}=\{\alpha_{\overline{v_{i'}}},\delta\}$,
$S^j_{v_i}=\{\alpha_{v_i},\beta^j_{v_i}\}$, and
$S^j_{\overline{v_i}}=\{\alpha_{\overline{v_i}},\beta^j_{\overline{v_i}}\}$.
So, $\phi_\sigma(S^j_{v_i})\cup\phi_\sigma(S^j_{\overline{v_i}})\subseteq {\cal
A}_i\cup {\cal B}_j$ whereas
$\phi_\sigma(H_{v_{i'}})\cup\phi_\sigma(H_{\overline{v_{i'}}})\subseteq {\cal
A}_{i'}\cup\{\delta\}$. Thus, since $i'<i\leq n$, we conclude that
$\phi_\sigma(S^j_{v_i})\cup\phi_\sigma(S^j_{\overline{v_i}})$ and
$\phi_\sigma(H_{v_{i'}})\cup\phi_\sigma(H_{\overline{v_{i'}}})$ are in different
components of $T_I-y_{i'}y_{i'+1}$. This proves that ${\cal T}_\sigma$ displays
all the four quartet trees $H_{v_{i'}}|S^j_{v_i}$,
$H_{\overline{v_{i'}}}|S^j_{v_i}$, $H_{v_{i'}}|S^j_{\overline{v_i}}$ and
$H_{\overline{v_{i'}}}|S^j_{\overline{v_i}}$.

Similarly, we consider $H_{\overline{v_i}}|F^j$ and $H_{v_i}|F^j$ for
$i\in\{1\ldots n\}$ and $j\in\{1\ldots m\}$. Recall that
$H_{v_i}=\{\alpha_{v_i},\delta\}$,
$H_{\overline{v_i}}=\{\alpha_{\overline{v_i}},\delta\}$, and
$F^j=\{\lambda^j,\mu\}$. Hence, it follows that
$\{\phi_\sigma(H_{\overline{v_i}})\cup \phi_\sigma(H_{v_i})\}\subseteq {\cal
A}_i\cup\{\delta\}$ and $\phi_\sigma(F^j)\subseteq {\cal B}_j\cup\{\mu\}$. Thus,
we conclude that $\phi_\sigma(H_{\overline{v_i}})\cup\phi_\sigma(H_{v_i})$ and
$\phi_\sigma(F^j)$ are in different components of $T_I-y_nu_1$. This proves that
${\cal T}_\sigma$ displays both $H_{\overline{v_i}}|F^j$ and $H_{v_i}|F^j$.

Finally, we consider the clause ${\cal C}_j=X\vee Y\vee Z$ for $j\in\{1\ldots
m\}$. Since $\sigma$ is a satisfying assignment, and by the rotational symmetry
between $X$, $Y$, and $Z$, we may assume that $X=1$, $Y=0$, and $Z=0$.  Let
$i_X$ be the index such that $X=v_{i_X}$ or $X=\overline{v_{i_X}}$, let $i_Y$ be
such that $Y=v_{i_Y}$ or $Y=\overline{v_{i_Y}}$, and let $i_Z$ be such that
$Z=v_{i_Z}$ or $Z=\overline{v_{i_Z}}$. Note that $i_X$, $i_Y$, $i_Z$ are
all distinct, since we assume that no variable appears more than once in each
clause. Thus we have that $\phi_\sigma(\beta^j_X)=b^j_1$,
$\phi_\sigma(\beta^j_{\overline Y})=b^j_2$, $\phi_\sigma(\beta^j_{\overline
Z})=b^j_3$, $\phi_\sigma(\gamma^j_1)=g^j_1$, $\phi_\sigma(\gamma^j_2)=g^j_2$,
$\phi_\sigma(\gamma^j_3)=g^j_3$, and \mbox{$\phi_\sigma(\lambda^j)=\ell_j$}.
(See Figure~\ref{fig:3a}c.) Also, 
$\{\phi_\sigma(\alpha_X),\phi_\sigma(\alpha_{\overline
X}),\phi_\sigma(\beta^j_{\overline X})\}\subseteq {\cal A}_{i_X}$,
$\{\phi_\sigma(\alpha_Y),\phi_\sigma(\alpha_{\overline
Y}),\phi_\sigma(\beta^j_Y)\}\subseteq{\cal A}_{i_Y}$, and
$\{\phi_\sigma(\alpha_Z),\phi_\sigma(\alpha_{\overline
Z}),\phi_\sigma(\beta^j_Z) \}\subseteq {\cal A}_{i_Z}$.
First, consider $K^j_{\overline X}|K^j_X$ and $K^j_{\overline X}|L^j_X$.  Recall
that $K^j_{\overline X}=\{\beta^j_X,\gamma^j_1\}$, $K^j_X=\{\beta^j_{\overline
X},\lambda^j\}$, and $L^j_X=\{\beta^j_{\overline X},\gamma^j_2\}$.  Also, recall
that $\phi_\sigma(\beta^j_{\overline X})\in{\cal A}_{i_X}$.  Thus it follows
that $\phi_\sigma(K^j_X)\cup\phi_\sigma(L^j_X)$ and $\phi_\sigma(K^j_{\overline
X})$ are in different components of $T_I-x^j_4 x^j_6$.
Now, consider $K^j_{\overline Y}|K^j_Y$ and $K^j_{\overline Y}|L^j_Y$. Recall
that $K^j_{\overline Y}=\{\beta^j_Y,\gamma^j_2\}$, $K^j_Y=\{\beta^j_{\overline
Y},\lambda^j\}$, and $L^j_Y=\{\beta^j_{\overline Y},\gamma^j_3\}$ where
$\phi_\sigma(\beta^j_Y)\in{\cal A}_{i_Y}$.  Thus,
$\phi_\sigma(K^j_Y)\cup\phi_\sigma(L^j_Y)$ and $\phi_\sigma(K^j_{\overline Y})$
are in different components of $T_I-x^j_1 x^j_2$.
Similarly, consider $K^j_{\overline Z}|K^j_Z$ and $K^j_{\overline Z}|L^j_Z$.
Recall that $K^j_{\overline Z}=\{\beta^j_Z,\gamma^j_3\}$,
$K^j_Z=\{\beta^j_{\overline Z},\lambda^j\}$, and $L^j_Z=\{\beta^j_{\overline
Z},\gamma^j_1\}$ where $\phi_\sigma(\beta^j_Z)\in{\cal A}_{i_Z}$.  Thus,
$\phi_\sigma(K^j_Z)\cup\phi_\sigma(L^j_Z)$ and $\phi_\sigma(K^j_{\overline Z})$
are in different components of $T_I-x^j_2 x^j_4$.
Now, consider $S^j_Y|K^j_X$ and $S^j_Y|L^j_Z$. Recall that
$S^j_Y=\{\alpha_Y,\beta^j_Y\}$, $K^j_X=\{\beta^j_{\overline X},\lambda^j\}$ and
$L^j_Z=\{\beta^j_{\overline Z},\gamma^j_1\}$.  Also,
$\{\phi_\sigma(\alpha_Y),\phi_\sigma(\beta^j_Y)\}\subseteq{\cal A}_{i_Y}$
whereas $\phi_\sigma(\beta^j_X)\in{\cal A}_{i_X}$.  Thus, since $i_X\neq i_Y$,
we conclude that $\phi_\sigma(S^j_Y)$ and $\phi_\sigma(K^j_X)\cup
\phi_\sigma(L^j_Z)$  are in different components of $T_I-y_{i_Y} y'_{i_Y}$.
Similarly, we consider $S^j_Z|K^j_Y$ and $S^j_Z|L^j_X$.  Recall that
$S^j_Z=\{\alpha_Z,\beta^j_Z\}$, $K^j_Y=\{\beta^j_{\overline Y},\lambda^j\}$, and
$L^j_X=\{\beta^j_{\overline X},\gamma^j_2\}$.  Also,
$\{\phi_\sigma(\alpha_Z),\phi_\sigma(\beta^j_Z)\}\subseteq{\cal A}_{i_Z}$, and
$\phi_\sigma(\beta^j_{\overline X})\in{\cal A}_{i_X}$.  Thus, since $i_X\neq
i_Z$, we conclude that $\phi_\sigma(S^j_Z)$ and
$\phi_\sigma(K^j_Y)\cup\phi_\sigma(L^j_X)$ are in different components of
$T_I-y_{i_Z} y'_{i_Z}$.
Finally, consider $S^j_X|K^j_Z$ and $S^j_X|L^j_Y$. Recall that
$S^j_X=\{\alpha_X,\beta^j_X\}$, $K^j_Z=\{\beta^j_{\overline Z},\lambda^j\}$ and
$L^j_Y=\{\beta^j_{\overline Y},\gamma^j_3\}$ where
$\phi_\sigma(\alpha_X)\in{\cal A}_{i_X}$. Thus, $\phi_\sigma(S^j_X)$ and
$\phi_\sigma(K^j_Z)$ are in different components of $T_I-x^j_4 x^j_5$, whereas
$\phi_\sigma(S^j_X)$ and $\phi_\sigma(L^j_Y)$ are in different components of
$T_I-x^j_2 x^j_3$.
\medskip

This proves that ${\cal T}_\sigma$ displays ${\cal Q}_I$.  It remains to prove
that ${\cal T}_\sigma$ is distinguished by ${\cal Q}_I$.  First, consider the
edge $y_iy'_i$ for $i\in\{1\ldots n\}$.  Recall that
$A_i=\{\alpha_{v_i},\alpha_{\overline{v_i}}\}$ and $B=\{\delta,\mu\}$. By
definition, we have $\phi_\sigma(A_i)=\{a_i,a'_i\}$ and
$\phi_\sigma(B)=\{y_0,u_0\}$. Note that every connected subgraph of $T_I$ that
contains both $y_0$ and $u_0$ must also contain $y_i$, since it lies on the path
between $u_0$ and $y_0$ in $T_I$. Likewise, every connected subgraph of $T_I$
that contains $a_i,a_i'$ also contains $y'_i$. Thus, this shows that the edge
$y_iy'_i$ is distinguished by $A_i|B$ which is in ${\cal Q}_I$.  We similarly
consider the edge $u_j x^j_1$ for $j\in\{1\ldots m\}$. By the definition of
$\phi_\sigma$, we observe that there exists $p\in\{1,2,3\}$ such that
$\phi_\sigma(\gamma^j_p)=g^j_2$.  We recall that $B=\{\delta,\mu\}$ and
$D^j_p=\{\gamma^j_p,\lambda^j\}$.  Thus, $\phi_\sigma(B)=\{y_0,u_0\}$ and
$\phi_\sigma(D^j_p)=\{g^j_2,\ell^j\}$. Since $g^2_j$ is adjacent to $x^j_1$, and
$u_j$ lies on the path between $y_0$ and $u_0$, it follows that the edge
$u_jx^j_1$ is distinguished by $D^j_p | B$ which is in ${\cal Q}_I$.

Now, consider $i\in\{1\ldots n\}$, and let $j_1<j_2<\ldots<j_t$ be the elements
of $\Delta_i$.  Let $W\in\{v_i,\overline{v_i}\}$ be such that $W=1$.  Then we
have $\phi_\sigma(\alpha_{W})=a_i$, $\phi_\sigma(\alpha_{\overline{W}})=a'_i$,
and $\phi_\sigma(\beta^j_{\overline{W}})=c^j_i$ for all $j\in\Delta_i$.  Recall
that $S^j_{\overline{W}}=\{\alpha_{\overline{W}},\beta^j_{\overline{W}}\}$ and
$K^j_{W}\subseteq\{\beta^j_{\overline{W}}, \gamma^j_1, \gamma^j_2, \gamma^j_3,
\lambda^j\}$  where $\{\phi_\sigma(\gamma^j_1), \phi_\sigma(\gamma^j_2),
\phi_\sigma(\gamma^j_3), \phi_\sigma(\lambda^j)\} \subseteq{\cal B}_j$ for all
$j\in\Delta_i$.  Thus, for each $k\in\{1\ldots t-1\}$, it follows that
$\phi_\sigma(\beta^{j_k}_{\overline{W}})$ is adjacent to $z^{j_k}_i$ whereas
$\phi_\sigma(\beta^{j_{k+1}}_{\overline{W}})$ is adjacent to $z^{j_{k+1}}_i$.
This proves that the edge $z^{j_k}_iz^{j_{k+1}}_i$ is distinguished by
$S^{j_k}_{\overline{W}}|K^{j_{k+1}}_{W}$. Similarly, recall that
$S^j_{W}=\{\alpha_{W},\beta^j_{W}\}$ where $\phi_\sigma(\beta^j_{W})\in{\cal B}_j$ and
$\phi_\sigma(\alpha_{W})$ is adjacent to $y'_i$. Thus, the edge $z^{j_t}_i y'_i$
is distinguished by $S^{j_t}_{W} | S^{j_t}_{\overline{W}}$.  Further, if $i\geq
2$, then we recall that $H_{v_{i-1}}=\{\alpha_{v_{i-1}},\delta\}$ where
$\phi_\sigma(\alpha_{v_{i-1}})\in{\cal A}_{i-1}$ and $\phi_\sigma(\delta)=y_0$.
Thus $y_{i-1}y_i$ is distinguished by $H_{v_{i-1}}|S^{j_t}_W$.

Now, consider $j\in\{1,\ldots m\}$ where ${\cal C}_j=X\vee Y\vee Z$.  By the
rotational symmetry, we may assume that $X=1$ and $Y=Z=0$.  Thus
\mbox{$\phi_\sigma(\beta^j_X)=b^j_1$}, \mbox{$\phi_\sigma(\beta^j_{\overline
Y})=b^j_2$}, $\phi_\sigma(\beta^j_{\overline Z})=b^j_3$,
$\phi_\sigma(\gamma^j_1)=g^j_1$, $\phi_\sigma(\gamma^j_2)=g^j_2$,
$\phi_\sigma(\gamma^j_3)=g^j_3$, and $\phi_\sigma(\lambda^j)=\ell_j$.  (Again
see Figure~\ref{fig:3a}c.) Recall that $K^j_Y=\{\beta^j_{\overline
Y},\lambda^j\}$ and $K^j_{\overline Y}=\{\beta^j_Y,\gamma^j_2\}$ where
$\phi_\sigma(\beta^j_Y)\not\in{\cal B}_j$.  This shows that the edge $x^j_1
x^j_2$ is distinguished by $K^j_{\overline Y}|K^j_Y$. Recall that
$S^j_X=\{\alpha_X,\beta^j_X\}$, $L^j_Y=\{\beta^j_{\overline Y},\gamma^j_3\}$,
and $K^j_Z=\{\beta^j_{\overline Z},\lambda^j\}$ where
$\phi_\sigma(\alpha_X)\not\in{\cal B}_j$.  Thus, the edge $x^j_2x^j_3$ is
distiguished by $S^j_X | L^j_Y$ whereas the edge $x^j_4 x^j_5$ is distinguished
by $S^j_X | K^j_Z$.  Recall that $K^j_{\overline Z}=\{\beta^j_Z,\gamma^j_3\}$
and $L^j_Z=\{ \beta^j_{\overline Z},\gamma^j_1\}$ where
$\phi_\sigma(\beta^j_Z)\not\in{\cal B}_j$.  Thus, the edge $x^j_2 x^j_4$ is
distinguished by $K^j_{\overline Z}|L^j_Z$.  Recall that
$K^j_X=\{\beta^j_{\overline X},\lambda^j\}$ and $K^j_{\overline
X}=\{\beta^j_X,\gamma^j_1\}$ where $\phi_\sigma(\beta^j_X)\not\in{\cal B}_j$.
Thus, the edge $x^j_4 x^j_6$ is distinguished by $K^j_{\overline X}|K^j_X$.
Further, if $j<m$, recall that $F^{j+1}=\{\lambda^{j+1},\mu\}$ where
$\phi_\sigma(\lambda^{j+1})\in{\cal B}_{j+1}$ and $\phi_\sigma(\mu)=u_0$. Thus
$u_ju_{j+1}$ is distinguished by $K^j_X | F^{j+1}$.

Finally, recall that $H_{v_n}=\{\alpha_{v_n},\delta\}$ and
$F^1=\{\lambda^1,\mu\}$. So, $\phi_\sigma(H_{v_n})\subseteq {\cal
A}_n\cup\{y_0\}$ and $\phi_\sigma(F^1)\subseteq{\cal B}_j\cup\{u_0\}$.  Thus,
the edge $y_nu_1$ is distinguished by $H_{v_n}|F^1$. 

This concludes the proof.
\end{proof}

\section{Proof of Theorem \ref{thm:cell-inter}}\label{sec:sandwich}
To prove Theorem \ref{thm:cell-inter}, we need to introduce some additional tools.
The following is a standard property of minimal chordal completions.

\begin{lemma}\label{lem:sand1}
Let $G'$ be a chordal completion of $G$. Then $G'$ is a minimal chordal
completion of $G$ if and only if for all $uv\in E(G')\setminus E(G)$, the
vertices $u,v$ have at least two non-adjacent common neighbours in $G'$.
\end{lemma}
\begin{proof}
Suppose that $G'$ is a minimal chordal completion.  Let $uv\in E(G')\setminus
E(G)$, and let $G''=G'-uv$. Since $G'$ is a minimal chordal completion and
$uv\not\in E(G)$, we conclude that $G''$ is not chordal.  Thus, there exists a
set $C\subseteq V(G')$ that induces a cycle in $G''$.  Since $G'$ is chordal,
$C$ does not induce a cycle in $G'$. This implies $u,v\in C$, and hence, $uv$ is
the unique chord of $G'[C]$. So, we conclude $|C|=4$, because otherwise $G'[C]$
contains an induced cycle. Let $x,y$ be the two vertices of $C\setminus\{u,v\}$.
Clearly, $xy\not\in E(G')$ and both $x$ and $y$ are common neighbours of $u,v$
as required.

Conversely, suppose that $G'$ is not a minimal chordal completion. Then by
\cite{rosetarjan}, there exists an edge $uv\in E(G')\setminus E(G)$ such that
$G'-uv$ is a chordal graph.  Therefore, $u,v$ do not have non-adjacent common
neighbours $x,y$ in $G'$, since otherwise $\{u,x,v,y\}$ induces a 4-cycle in
$G'-uv$, which is impossible since we assume that $G'-uv$ is chordal. That
concludes the proof.
\end{proof}

Using this tool, we prove the following two important lemmas.

\begin{lemma}\label{lem:sand2}
Let $G$ be a graph and $G'$ be a minimal chordal completion of $G$.
If $G$ contains vertices $u,v$ with $N_G(u)\subseteq N_G(v)$, then
also $N_{G'}(u)\subseteq N_{G'}(v)$.
\end{lemma}

\begin{proof}
Let $u,v$ be vertices of $G$ with $N_G(u)\subseteq N(G_v)$. Let
$B=N_{G'}(u)\setminus N_{G'}(v)$ and $A=N_{G'}(u)\cap N_{G'}(v)$. Assume that
$B\neq\emptyset$, and let $A_1$ denote the vertices of $A$ with at least one
neighbour in $B$.  Look at the graph $G_1=G'[A_1\cup B\cup\{v\}]$.

By the definition of $A_1$ and $B$, the vertex $v$ is adjacent to each vertex of
$A_1$ and non-adjacent to each vertex of $B$. Hence, no vertex of $A_1$ is
simplicial in $G_1$, since it is adjacent to $v$ and at least one vertex in $B$.

Now, consider $w\in B$. By the definition of $B$, we have that $w$ is adjacent
in $G'$ to $u$ but not $v$. Thus, $uw$ is not an edge of $G$, since
$N_G(u)\subseteq N_G(v)$ and $N_{G}(v)\subseteq N_{G'}(v)$.  So, by Lemma
\ref{lem:sand1}, the vertices $u,w$ have non-adjacent common neighbours $x,y$ in
$G'$. Since $x,y$ are adjacent to $u$, we have $x,y\in A\cup B$.  In fact, since
$w$ has no neighbours in $A\setminus A_1$, we conclude $x,y\in A_1\cup B$. Thus,
$w$ is not a simplicial vertex in $G_1$, and hence, no vertex of $B$ is
simplicial~in~$G_1$.

This proves that no vertex of $G_1$, except possibly for $v$, is simplicial in
$G_1$.  Also, $G_1$ is not a complete graph, since $B\neq\emptyset$, and $v$ has no
neighbour in $B$. Recall that $G_1$ is chordal because $G'$ is. Thus, by the
result of Dirac \cite{dirac}, $G_1$ must contain at least two non-adjacent
simplicial vertices, but that is impossible.  Hence, we must conclude
$B=\emptyset$. In other words, $N_{G'}(u)\subseteq N_{G'}(v)$.
\end{proof}

\begin{lemma}\label{lem:sand3}
Let $G$ be a graph, and let $H$ be a graph obtained from $G$ by substituting
complete graphs for the vertices of $G$. Then there is a one-to-one
correspondence between minimal chordal completions of $G$ and $H$.
\end{lemma}
\begin{proof}
Let $v_1$,$v_2$,\ldots,$v_n$ be the vertices of $G$.  Since $H$ is obtained from
$G$ by substituting complete graphs, there is a partition $C_1\cup\ldots\cup C_n$
of $V(H)$ where each $C_i$ induces a complete graph in $H$, and
for every distinct $i,j\in\{1\ldots n\}$:\medskip

\noindent($\star$) each $x\in C_i$, $y\in C_j$ satisfy $v_iv_j\in E(G)$ if and
only if $xy\in E(H)$.\medskip

Let $G'$ be any graph with vertex set $V(G)$, and let $H'=\Psi(G')$ be the graph
constructed from $G'$ by, for each $i\in\{1\ldots n\}$, substituting $C_i$ for
the vertex $v_i$, and making $C_i$ into a complete graph. Thus, for every
distinct $i,j\in\{1\ldots n\}$ \medskip

\noindent($\star\star$) each $x\in C_i$, $y\in C_j$ satisfy $v_iv_j\in E(G')$ if
and only if $xy\in E(H')$.\medskip

We prove that $\Psi$ is a bijection between the minimal chordal completions of $G$
and $H$ which will yield the claim of the lemma.

Let $G'$ be a minimal chordal completion of $G$, and let $H'=\Psi(G')$.
Clearly, $H'$ is chordal, since $G'$ is chordal, and chordal graphs are closed
under the operation of substituting a complete graph for a vertex.  Also,
observe that $V(H)=V(H')$, and if $xy\in E(H)$, then either $x,y\in C_i$ for
some $i\in\{1\ldots n\}$, in which case $xy\in E(H')$, since $C_i$ induces a
complete graph in $H'$, or we have $x\in C_i$, $y\in C_j$ for distinct
$i,j\in\{1\ldots n\}$ in which case $v_iv_j\in E(G)$ by ($\star$) implying
$v_iv_j\in E(G')$, since $E(G)\subseteq E(G')$, and hence, $xy\in E(H')$ by
($\star\star$).  This proves that $E(H)\subseteq E(H')$, and thus, $H'$ is a
chordal completion of $H$.

To prove that $H'$ is a minimal chordal completion of $H$, it suffices, by
Lemma~\ref{lem:sand1}, to show that for all $xy\in E(H')\setminus E(H)$, the
vertices $x,y$ have at least two non-adjacent common neighbours in $H'$.
Consider $xy\in E(H')\setminus E(H)$, and let $i,j\in\{1\ldots n\}$ be such that
$x\in C_i$ and $y\in C_j$.  Since $xy\not\in E(H)$ and $C_i$ induces a complete
graph in $H$, we conclude $i\neq j$. Thus, by ($\star\star$), we have $v_iv_j\in
E(G')$, and so, $v_iv_j\in E(G')\setminus E(G)$ by ($\star$). Now, recall that
$G'$ is a minimal chordal completion of $G$. Thus, by Lemma \ref{lem:sand1}, the
vertices $v_i$, $v_j$ have non-adjacent common neighbours $v_k$, $v_{\ell}$ in
$G'$.  So, we let $w\in C_k$ and $z\in C_\ell$. By ($\star\star$), we conclude
$wz\not\in E(H')$, since $v_kv_\ell\not\in E(G')$.  Moreover, ($\star\star$)
also implies that $z,w$ are common neighbours of $x,y$, since $v_k$, $v_\ell$
are common neighbours of $v_i$, $v_j$. This proves that $x,y$ have non-adjacent
common neighbours, and thus shows that $H'$ is a minimal chordal completion of $H$.

Conversely, let $H'$ be a minimal chordal completion of $H$. Let $G'$ be the
graph with $V(G')=V(G)$ such that $v_iv_j\in E(G')$ if and only if there exists
$x\in C_i$, $y\in C_j$ with $xy\in E(H')$.  Let $i\in\{1\ldots n\}$ and consider
vertices $x,y\in C_i$ in the graph $H$.  Recall that $C_i$ induces a complete
graph in $H$.  This implies that $xy\in E(H)$ and both $x$ and $y$ are adjacent
in $H$ to every $z\in C_i\setminus\{x,y\}$. Further, by ($\star$), if $z\in C_j$
where $j\neq i$, then $x,y$ are both adjacent to $z$ if $v_iv_j\in E(G)$, and
$x,y$ are both non-adjacent to $z$ if $v_iv_j\not\in E(G)$.  This shows that
$N_H(x)=N_H(y)$, and hence, $N_{H'}(x)=N_{H'}(y)$ by Lemma \ref{lem:sand2} and
the fact that $H'$ is a minimal chordal completion of $H$. This proves that
$H'=\Psi(G')$, and hence, $G'$ is chordal. In fact,  $E(G)\subseteq E(G')$ by
($\star$) and ($\star\star$). Thus $G'$ is a chordal completion of $G$.

It remains to show that $G'$ is a minimal chordal completion of~$G$. Again, it
suffices to show that for each $v_iv_j\in E(G')\setminus E(G)$, the vertices
$v_i,v_j$ have non-adjacent common neighbours in $G'$.  Consider $v_iv_j\in
E(G')\setminus E(G)$, and let $x\in C_i$ and $y\in C_j$.  So, $i\neq j$ and
$xy\in E(H')$ by ($\star\star$). Further, $xy\in E(H')\setminus E(H)$ by
($\star$) and the fact that $v_iv_j\not\in E(G)$. So, the vertices $x,y$ have
non-adjacent common neighbours $w,z$ in $H'$ by Lemma \ref{lem:sand2} and the
fact that $H'$ is a minimal chordal completion of $H$.  Let $k,\ell\in\{1\ldots
n\}$ be such that $w\in C_k$ and $z\in C_\ell$.  Since $xz\in E(H')$ but
$wx\not\in E(H')$, we conclude by ($\star\star$) that $i\neq k$. By symmetry,
also $i\neq \ell$, $j\neq k$, and $j\neq \ell$.  Further, $k\neq\ell$, since
$wx\not\in E(H')$ and $C_k$ induces a complete graph in $H'$.  Thus,
($\star\star$) implies that $v_k$, $v_\ell$ are non-adjacent common neighbours
of $v_i$, $v_j$, since $w,z$ are non-adjacent common neighbours of $x,y$. This
proves that $G'$ is indeed a minimal chordal completion of $G$.

That concludes the proof.
\end{proof}

Now, we are finally ready to prove Theorem \ref{thm:cell-inter}.\medskip

\begin{proofof}{Theorem \ref{thm:cell-inter}}
We observe that the graph ${\rm int}({\cal Q})$ is obtained by substituting
complete graphs for the vertices of ${\rm int^*}({\cal Q})$.  Thus, by Lemma
\ref{lem:sand3}, there is a bijection $\Psi$ between the minimal chordal
completions of ${\rm int}({\cal Q})$ and ${\rm int^*}({\cal Q})$.

By translating the condition ($\star\star$) from the proof of Lemma
\ref{lem:sand3}, we obtain that if $G'$ is a minimal chordal completion of ${\rm
int^*}({\cal Q})$, then $H'=\Psi(G')$ is the graph whose vertex set is that of
${\rm int}({\cal Q})$ with the property that for all $A,A'\in V(G')$\medskip

\noindent($\star\star$) all meaningful $\pi,\pi'\in {\cal Q}$ satisfy $AA'\in
V(G')$ $\iff$ \mbox{$(A,\pi)(A',\pi')\in V(H')$}.\medskip

We show that $\Psi$ is a bijection between the minimal restricted chordal
completions of ${\rm int}({\cal Q})$ and the minimal chordal sandwiches of $({\rm
int^*}({\cal Q}),{\rm forb}({\cal Q}))$.

First, let $H'$ be a minimal restricted chordal completion of ${\rm int}({\cal
Q})$.  Then $G'=\Psi^{-1}(H')$ is a minimal chordal completion of ${\rm
int^*}({\cal Q})$. Consider two cells $A_1$,\,$A_2$ of $\pi\in {\cal Q}$. Since
$H'$ is a restricted chordal completion of ${\rm int}({\cal Q})$, we have that
$(A_1,\pi)$ is not adjacent to $(A_2,\pi)$ in $H'$. Thus, $A_1A_2\not\in E(G')$
by ($\star\star$). This shows that $G'$ contains no edge of ${\rm forb}({\cal
Q})$. Thus $G'$ is a minimal chordal sandwich of $({\rm int^*}({\cal Q}),{\rm
forb}({\cal Q}))$, since it is also a minimal chordal completion~of~${\rm int^*}({\cal Q})$.

Conversely, let $G'$ be a minimal chordal sandwich of $({\rm int^*}({\cal Q})$,
${\rm forb}({\cal Q}))$. Then $H'=\Psi(G')$ is a minimal chordal completion of
${\rm int}({\cal Q})$. Consider two cells $A_1$,\,$A_2$ of $\pi\in{\cal Q}$.
Since $A_1A_2$ is an edge of ${\rm forb}({\cal Q})$, and $G'$ is a minimal
chordal sandwich of $({\rm int^*}({\cal Q})$,${\rm forb}({\cal Q}))$, we have
$A_1A_2\not\in E(G')$.  Thus, $(A_1,\pi)(A_2,\pi)\not\in E(H')$ by
($\star\star$). This shows that $H'$ is a minimal restricted chordal completion
of ${\rm int}({\cal Q})$.

That concludes the proof.
\end{proofof}

\section{Proof of Theorem \ref{thm:one}}\label{sec:proof-one}

For the proof, we shall need the following simple properties of chordal graphs.

\begin{lemma}\label{lem:4cycle}
Let $G$ be a chordal graph, and let $a,b$ be non-adjacent vertices of $G$. Then
every two common neighbours of $a$ and $b$ are adjacent.
\end{lemma}

\begin{lemma}\label{lem:5cycle}
Let $G$ be a chordal graph, and $C=\{a,b,c,d,e\}$ be a 5-cycle in $G$
with edges $ab,bc,cd,de,ae$.  Then
\vspace{-1.5ex}
\begin{enumerate}[(a)]
\item $bd,ce\not\in E(G)$ implies $ac,ad\in E(G)$, and
\item $bd,be\not\in E(G)$ implies $ac\in E(G)$.
\end{enumerate}
\vspace{-2ex}
\end{lemma}

\begin{lemma}\label{lem:6cycle}
Let $G$ be a chordal graph, and $C=\{a,b,c,d,e,f\}$ be a 6-cycle in $G$ with
edges $ab,bc,cd,de,ef,af$. Then\vspace{-1.5ex}
\begin{enumerate}[(a)]
\item $bd,ce,df\not\in E(G)$ implies $ac,ad,ae\in E(G)$,
\item $bd,ce,cf\not\in E(G)$ implies $ac,ad\in E(G)$, and
\item $be,bf,ce,cf\not\in E(G)$ implies $ad\in E(G)$.
\end{enumerate}
\vspace{-2ex}
\end{lemma}

To assist the reader in following the subsequent arguments, we list here the
cliques of ${\rm int^*}({\cal Q}_I)$ according to the elements from which they
arise:

\begin{enumerate}[a)]
\item for each $i\in\{1\ldots n\}$ where $j_1,j_2,\ldots,j_k$ are the elements of $\Delta_i$:
\vspace{0.5ex}

$\alpha_{v_i}$: $H_{v_i}$, $A_i$, $S^{j_1}_{v_i}$, $S^{j_2}_{v_i}$, \ldots,
$S^{j_t}_{v_i}$,
\hfill
$\alpha_{\overline{v_i}}$: $H_{\overline{v_i}}$, $A_i$,
$S^{j_1}_{\overline{v_i}}$, $S^{j_2}_{\overline{v_i}}$, \ldots,
$S^{j_t}_{\overline{v_i}}$,
\vspace{0.5ex}

\item for each $j\in\{1\ldots m\}$ where ${\cal C}_j=X\vee Y\vee Z$:
\vspace{0.5ex}

$\lambda^j$: $K^j_X$, $K^j_Y$, $K^j_Z$, $D^j_1$, $D^j_2$, $D^j_3$, $F^j$

\begin{tabular}{@{}lll}
$\gamma^j_1$: $K^j_{\overline X}$, $L^j_Z$, $D^j_1$
&
$\gamma^j_2$: $K^j_{\overline Y}$, $L^j_X$, $D^j_2$
&
$\gamma^j_3$: $K^j_{\overline Z}$, $L^j_Y$, $D^j_3$
\vspace{0.5ex}\\

$\beta^j_X$: $S^j_X$, $K^j_{\overline X}$
&
$\beta^j_Y$: $S^j_Y$, $K^j_{\overline Y}$
&
$\beta^j_Z$: $S^j_Z$, $K^j_{\overline Z}$
\vspace{0.5ex}\\

$\beta^j_{\overline X}$: $S^j_{\overline X}$, $K^j_X$,
$L^j_X$\quad\quad\quad\quad
&
$\beta^j_{\overline Y}$: $S^j_{\overline Y}$, $K^j_Y$,
$L^j_Y$\quad\quad\quad\quad
&
$\beta^j_{\overline Z}$: $S^j_{\overline Z}$, $K^j_Z$, $L^j_Z$
\vspace{0.5ex}
\end{tabular}

\item
$\delta$: $B$, $H_{v_1}$, \ldots, $H_{v_n}$, $H_{\overline{v_1}}$, \ldots,
$H_{\overline{v_n}}$

$\mu$: $B$, $F^1$, \ldots, $F^m$
\end{enumerate}

We start with a useful lemma describing an important property of ${\rm
int^*}({\cal Q}_I)$.

\begin{lemma}\label{lem:star}
Let $G'$ be a chordal sandwich of $({\rm int^*}({\cal Q}_I),{\rm forb}({\cal
Q}_I))$,  and $1\leq i\leq n$.\vspace{-1.5ex}

\begin{enumerate}[(a)]
\item there is $W\in\{v_i,\overline{v_i}\}$ such that for all
$j\in\Delta_i$, $K^j_W$ is adjacent to $B$.
\item for each $j\in\Delta_i$, and each $W\in\{v_i,\overline{v_i}\}$, if $K^j_W$
is adjacent to $B$, then the vertices $S^j_W$, $K^j_W$, $L^j_W$ (if exists) are
adjacent to $B$, $A_i$, $H_W$, $H_{\overline W}$, $F^j$.
\end{enumerate}
\end{lemma}

\begin{proof}
Let $i\in\{1\ldots n\}$. First, we observe the following.
\medskip

\noindent($\star$)~{\em for each $j\in\Delta_i$, each
$W\in\{v_i,\overline{v_i}\}$, at least one of $S^j_{\overline W}$, $K^j_W$ is
adjacent to~$B$.}\medskip

\noindent We may assume that $S^j_{\overline W}$ is not adjacent to $B$,
otherwise we are done. Observe that $S^j_{\overline W}$ is adjacent to $K^j_W$,
since $\beta^j_{\overline W}\in K^j_W\cap S^j_{\overline W}$.  Moreover, there
exists $p\in\{1,2,3\}$ such that $K^j_W\cap D^j_p$ contains $\lambda^j$ or
$\gamma^j_p$, implying that $K^j_W$ is adjacent to $D^j_p$. Also, $F^j$ is
adjacent to $D^j_p$ and $B$, since $\lambda^j\in D^j_p\cap F^j$ and $\mu\in
B\cap F^j$, respectively. Further, $H_{\overline W}$ is adjacent to
$S^j_{\overline W}$ and $B$, since $\alpha_{\overline W}\in H_{\overline W}\cap
S^j_{\overline W}$ and $\delta\in H_{\overline W}\cap B$. Finally, $H_{\overline
W}$ is not adjacent to $F^j$, and $B$ is not adjacent to $D^j_p$, since
$H_{\overline W}|F^j$ and $D^j_p|B$ are in ${\cal Q}_I$. So, by Lemma
\ref{lem:6cycle} applied to the cycle $\{K^j_W$, $S^j_{\overline W}$,
$H_{\overline W}$, $B$, $F^j$, $D^j_p\}$, we conclude that $K^j_W$ is adjacent
to $B$. This proves ($\star$).  \medskip

Now, to prove (a), suppose for contradiction that there are $j,j'\in\Delta_i$
such that both $K^j_{\overline{v_i}}$ and $K^{j'}_{v_i}$ are not adjacent to
$B$.  Then by ($\star$), both $S^j_{v_i}$ and $S^{j'}_{\overline{v_i}}$ are
adjacent to $B$. Note also that $A_i$ is adjacent to both $S^j_{v_i}$,
$S^{j'}_{\overline{v_i}}$, since $\alpha_{v_i}\in A_i\cap S^j_{v_i}$ and
$\alpha_{\overline{v_i}}\in A_i\cap S^{j'}_{\overline{v_i}}$.  Further, note
that $A_i B$ and $S^j_{v_i} S^{j'}_{\overline{v_i}}$ are not edges of $G'$,
since $A_i|B$ and $S^j_{v_i}|S^{j'}_{\overline{v_i}}$ are in ${\cal Q}_I$. But
then $G'$ contains an induced 4-cycle on $\{S^j_{v_i}$, $A_i$,
$S^{j'}_{\overline{v_i}}$, $B\}$, which is impossible, since $G'$ is chordal.
This proves (a).

For (b), suppose that $K^j_W$ is adjacent to $B$ for $j\in\Delta_i$ and
$W\in\{v_i,\overline{v_i}\}$.  First observe that $K^j_W$ is adjacent to
$S^j_{\overline W}$, and the vertex $K^j_{\overline W}$ is adjacent to $S^j_W$,
since $\beta^j_{\overline W}\in K^j_W\cap S^j_{\overline W}$ and $\beta^j_W\in
K^j_{\overline W}\cap S^j_W$. Moreover, there exists $p\in\{1,2,3\}$ such that
$K^j_W\cap D^j_p$ and $K^j_{\overline W}\cap D^j_p$ contain $\gamma^j_p$ and
$\lambda^j$, respectively, or $\lambda^j$ and $\gamma^j_p$, respectively.  This
implies that $K^j_W$ and $K^j_{\overline W}$ are adjacent to $D^j_p$.  Also,
$A_i$ is adjacent to $S^j_W$ and $S^j_{\overline W}$, since $\alpha_W\in A_i\cap
S^j_W$ and $\alpha_{\overline W}\in A_i\cap S^j_{\overline W}$.  Further, note
that $D^j_p B$, $A_i B$, $K^j_W K^j_{\overline W}$, and $S^j_W S^j_{\overline
W}$ are not edges of $G'$, since $D^j_p|B$, $A_i|B$, $K^j_W|K^j_{\overline W}$,
and $S^j_W|S^j_{\overline W}$ are in ${\cal Q}_I$. This implies that
$K^j_{\overline W}$ is not adjacent to $B$, since otherwise $G'$ contains an
induced 4-cycle on $\{K^j_W$, $B$, $K^j_{\overline W}$, $D^j_p\}$. So, by
($\star$), we have that $S^j_W$ is adjacent to $B$. Thus, Lemma \ref{lem:5cycle}
applied to $\{K^j_W$, $S^j_{\overline W}$, $A_i$, $S^j_W$, $B\}$ yields that
$K^j_W$ is adjacent to $A_i$ and $S^j_W$. So, by Lemma \ref{lem:4cycle} applied
to $\{S^j_W$, $K^j_W$, $D^j_p$, $K^j_{\overline W}\}$, we have that $S^j_W$ is
adjacent to $D^j_p$.

Now, observe that $H_W$, $H_{\overline W}$ are adjacent to both $A_i$ and $B$,
since $\alpha_W\in H_W\cap A_i$, $\alpha_{\overline W}\in H_{\overline W}\cap
A_i$, and $\delta\in B\cap H_W\cap H_{\overline W}$.  Thus, by Lemma
\ref{lem:4cycle} applied to $\{u$, $A_i$, $u'$, $B\}$ where $u\in\{S^j_W$,
$K^j_W$$\}$ and $u'\in\{H_W,H_{\overline W}\}$ , we conclude
that $S^j_W$ and $K^j_W$  are adjacent to both $H_W$ and
$H_{\overline W}$.  Similarly, we observe that $F^j$ is adjacent to $B$ and
$D^j_p$, since $\mu\in F^j\cap B$ and $\lambda^j\in D^j_p\cap F^j$. Thus, Lemma
\ref{lem:4cycle} applied to $\{u$, $B$, $F^j$, $D^j_p\}$ yields that $S^j_W$
and $K^j_W$ are also adjacent to $F^j$.

Lastly, suppose that $L^j_W$ exists. Then there exists $q\in\{1,2,3\}$ such
that $\gamma^j_{q}\in D^j_{q}\cap L^j_W$ implying that $L^j_W$ is adjacent to
$D^j_{q}$.  Moreover, $F^j$ is adjacent to $D^j_{q}$ and $B$, since
$\lambda^j\in D^j_{q}\cap F^j$ and $\mu\in F^j\cap B$. Also, $H_{\overline W}$
is adjacent to $B$, $S^j_{\overline W}$, and the vertex $S^j_{\overline W}$ is
adjacent to $L^j_W$, since $\delta\in B\cap H_{\overline W}$, $\alpha_{\overline
W}\in H_{\overline W}\cap S^j_{\overline W}$, and $\beta^j_{\overline W}\in
S^j_{\overline W}\cap L^j_W$.  Further, $H_{\overline W} F^j$ and $D^j_{q}B$
are not edges of $G'$, since $H_{\overline W}|F^j$ and $D^j_{q}|B$ are in
${\cal Q}_I$.  Also, $S^j_{\overline W} B$ is not an edge of $G'$, since
otherwise $G'$ contains an induced 4-cycle on $\{S^j_W$, $B$, $S^j_{\overline
W}$, $A_i\}$.  Thus, by Lemma \ref{lem:5cycle} applied to $\{L^j_W$,
$S^j_{\overline W}$, $H_{\overline W}$, $B$, $F^j$, $D^j_{q}\}$, we conclude
that $L^j_W$ is adjacent to $H_{\overline W}$, $B$, and $F^j$.  Moreover,
by Lemma \ref{lem:5cycle} applied to $\{L^j_W$, $B$, $S^j_W$, $A_i$,
$S^j_{\overline W}\}$, we conclude that $L^j_W$ is adjacent to $A_i$.
Finally, recall that $H_W$ is adjacent to both $A_i$ and $B$. Thus, Lemma
\ref{lem:4cycle} applied to $\{L^j_W$, $A_i$, $H_W$, $B\}$ yields that $L^j_W$
is also adjacent to $H_W$.

That concludes the proof.
\end{proof}

\begin{figure}[h!t]
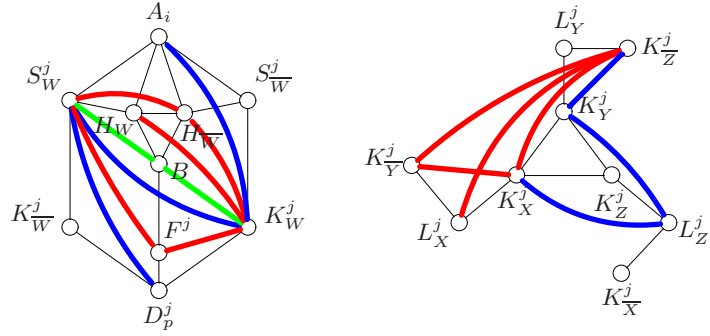

\centering
\vskip -1ex

\small 
\raisebox{-2ex}{$\xy/r4pc/:
(-0.7,-0.5)*[o][F]{\phantom{S}}="k'";
(0.7,-0.5)*[o][F]{\phantom{S}}="k";
(-0.7,0.5)*[o][F]{\phantom{S}}="s";
(0.7,0.5)*[o][F]{\phantom{S}}="s'";
(0,-1)*[o][F]{\phantom{S}}="d";
(0,1)*[o][F]{\phantom{S}}="a";
(0,0)*[o][F]{\phantom{S}}="b";
(-0.2,0.4)*[o][F]{\phantom{S}}="h";
(0.2,0.4)*[o][F]{\phantom{S}}="h'";
(0,-0.7)*[o][F]{\phantom{S}}="f";
{\ar@{-} "k";"d"};
{\ar@{-} "k'";"d"};
{\ar@{-} "s";"a"};
{\ar@{-} "s'";"a"};
{\ar@{-} "s";"k'"};
{\ar@{-} "s'";"k"};
{\ar@{-} "s";"h"};
{\ar@{-} "h";"h'"};
{\ar@{-} "h'";"s'"};
{\ar@{-} "h'";"a"};
{\ar@{-} "h";"a"};
{\ar@{-} "h";"b"};
{\ar@{-} "h'";"b"};
{\ar@{-} "f";"b"};
{\ar@{-} "f";"d"};
{\ar@{-}@*{[blue]}@*{[|2pt]}@/_0.8pc/ "s";"k"};
{\ar@{-}@*{[green]}@*{[|2pt]} "s";"b"};
{\ar@{-}@*{[green]}@*{[|2pt]} "k";"b"};
{\ar@{-}@*{[red]}@*{[|2pt]}@/_0.3pc/ "k";"h"};
{\ar@{-}@*{[red]}@*{[|2pt]}@/_0.3pc/ "k";"h'"};
{\ar@{-}@*{[blue]}@*{[|2pt]}@/_0.7pc/ "k";"a"};
{\ar@{-}@*{[red]}@*{[|2pt]} "k";"f"};
{\ar@{-}@*{[blue]}@*{[|2pt]}@/_0.5pc/ "s";"d"};
{\ar@{-}@*{[red]}@*{[|2pt]}@/_0.2pc/ "s";"f"};
{\ar@{-}@*{[red]}@*{[|2pt]}@/^0.4pc/ "s";"h'"};
"a"+(0,0.2)*{A_i};
"d"+(0,-0.2)*{D^j_p};
"s"+(-0.2,0.2)*{S^j_W};
"s'"+(0.2,0.2)*{S^j_{\overline W}};
"h"+(-0.15,-0.1)*{H_W};
"h'"+(0.13,-0.18)*{H_{\overline W}};
"f"+(0.15,0.2)*{F^j};
"b"+(0.16,-0.05)*{B};
"k"+(0.3,0.1)*{K^j_W};
"k'"+(-0.3,0.1)*{K^j_{\overline W}};
\endxy$}
\qquad$\xy/r2pc/:
(-0.75,-0.5)*[o][F]{\phantom{S}}="x";
(-1.65,-1.25)*[o][F]{\phantom{S}}="lx";
(0,0.5)*[o][F]{\phantom{S}}="y";
(0,1.5)*[o][F]{\phantom{S}}="ly";
(0.75,-0.5)*[o][F]{\phantom{S}}="z";
(1.65,-1.25)*[o][F]{\phantom{S}}="lz";
(0.9,-2.05)*[o][F]{\phantom{S}}="x'";
(-2.4,-0.35)*[o][F]{\phantom{S}}="y'";
(1,1.5)*[o][F]{\phantom{S}}="z'";
{\ar@{-} "x";"lx"};
{\ar@{-} "lx";"y'"};
{\ar@{-} "y";"ly"};
{\ar@{-} "ly";"z'"};
{\ar@{-} "z";"lz"};
{\ar@{-} "lz";"x'"};
{\ar@{-} "x";"y"};
{\ar@{-} "y";"z"};
{\ar@{-} "z";"x"};
"x'"+(0,-0.4)*{K^j_{\overline X}};
"y'"+(-0.45,0.1)*{K^j_{\overline Y}};
"z'"+(0.5,0)*{K^j_{\overline Z}};
"lx"+(-0.4,-0.15)*{L^j_X};
"x"+(0,-0.4)*{K^j_X};
"ly"+(0.1,0.45)*{L^j_Y};
"y"+(0.5,0.1)*{K^j_Y};
"lz"+(0.4,-0.1)*{L^j_Z};
"z"+(0,-0.45)*{K^j_Z};
{\ar@{-}@*{[red]}@*{[|2pt]} "x";"y'"};
{\ar@{-}@*{[red]}@*{[|2pt]}@/^0.7pc/ "x";"z'"};
{\ar@{-}@*{[red]}@*{[|2pt]}@/^0.5pc/ "y'";"z'"};
{\ar@{-}@*{[red]}@*{[|2pt]}@/^1.05pc/ "lx";"z'"};
{\ar@{-}@*{[blue]}@*{[|2pt]}@/_0.6pc/ "x";"lz"};
{\ar@{-}@*{[blue]}@*{[|2pt]} "y";"z'"};
{\ar@{-}@*{[blue]}@*{[|2pt]}@/^0.3pc/ "y";"lz"};
\endxy$
\caption{The fill-in edges for {\em a)} $W=1$, {\em b)} $X=1$, $Y=0$, $Z=0$.\label{fig:2}}
\end{figure}

Let $\sigma$ be a truth assignment for the instance $I$.  Recall that, for
simplicity, we write $X=0$ and $X=1$ in place of $\sigma(X)=0$ and
$\sigma(X)=1$, respectively.

To facilitate the arguments in the proof, we introduce a naming convention for
the vertices in ${\rm int^*}({\cal Q}_I)$ similar to that of \cite{twostrikes}.
The vertices $S^j_W$ for all meaningful choices of $j$ and $W$ are called {\em
shoulders}.  For a fixed $j$, we call them {\em shoulders of the clause ${\cal
C}_j$}, and for a fixed $W$, we call them {\em shoulders of the literal $W$}.  A
shoulder is a {\em a true shoulder} if $W=1$. Otherwise, it is a {\em false
shoulder}.  The vertices $K^j_W$, $L^j_W$ for all meaningful choices of $j$ and
$W$ are called {\em knees}. For a fixed $j$, we call them {\em knees of the
clause ${\cal C}_j$}, and for a fixed $W$, we call them {\em knees of the
literal $W$}.  A knee is a {\em true knee} if $W=1$. Otherwise, it is a {\em
false knee}.  The vertices $A_i$, $D^j_p$, $H_W$, $F^j$ for all meaningful
choices of indices are called {\em $A$-vertices}, {\em $D$-vertices}, {\em
$H$-vertices}, and {\em $F$-vertices}, respectively.

\medskip

Let $G_{\sigma}$ be the graph constructed from ${\rm int^*}({\cal Q}_I)$ by
performing the following:
\begin{enumerate}[(i)]
\item\label{enum:1}
make $B$ adjacent to all true knees and true shoulders
\end{enumerate}

Let $G'_\sigma$ be the graph constructed from $G_\sigma$ by performing the
following steps:
\begin{enumerate}[(i)]
\setcounter{enumi}{1}
\item\label{enum:2}
make $\{$true knees, true shoulders$\}$ into a complete graph

\item\label{enum:3}
for all $i\in\{1\ldots n\}$, make $A_i$ adjacent to all true knees of the
literals $v_i$,$\overline{v_i}$,

\item\label{enum:4}
for all $1\leq i'\leq i\leq n$, make $H_{v_i}$, $H_{\overline{v_i}}$ adjacent
to all true knees and true shoulders of the literals $v_{i'}$, $\overline{v_{i'}}$

\item\label{enum:5}
for all $1\leq j\leq j'\leq m$, make $F^j$ adjacent to all true knees and
true shoulders of the clause ${\cal C}_{j'}$,

\item\label{enum:6}
for all $1\leq i\leq n$ and all $j,j'\in\Delta_i$ such that $j\leq j'$:
\begin{enumerate}[a)]
\item if $v_i=1$, make $S^{j'}_{\overline{v_i}}$ adjacent to $K^j_{v_i}$,
$L^j_{v_i}$ (if exists)

\item if $v_i=0$, make $S^{j'}_{v_i}$ adjacent to $K^j_{\overline{v_i}}$,
$L^j_{\overline{v_i}}$ (if exists)
\end{enumerate}

\end{enumerate}

Finally, let $G^*_\sigma$ be constructed from $G'_\sigma$ by adding the following edges.

\begin{enumerate}[(i)]
\setcounter{enumi}{6}
\item\label{enum:7}
for all $j\in\{1\ldots m\}$ where ${\cal C}_j=X\vee Y\vee Z$:
\begin{enumerate}[a)]
\item if $X=1$, then add edges $F^j L^j_Z$, $K^j_X L^j_Z$, $K^j_Y K^j_{\overline
Z}$, $D^j_2 K^j_{\overline Z}$, $D^j_2 S^j_{\overline Y}$, $D^j_3 S^j_{\overline
Y}$ and make $\{D^j_1$, $D^j_2$, $D^j_3$, $S^j_X$, $S^j_{\overline Z}$, $L^j_Z$,
$K^j_Y\}$ into a complete graph

\item if $Y=1$, then add edges $F^j L^j_X$, $K^j_Y L^j_X$, $K^j_Z K^j_{\overline
X}$, $D^j_3 K^j_{\overline X}$, $D^j_3 S^j_{\overline Z}$, $D^j_1 S^j_{\overline
Z}$ and make $\{D^j_1$, $D^j_2$, $D^j_3$, $S^j_Y$, $S^j_{\overline X}$, $L^j_X$,
$K^j_Z\}$ into a complete graph

\item if $Z=1$, then add edges $F^j L^j_Y$, $K^j_Z L^j_Y$, $K^j_X K^j_{\overline
Y}$, $D^j_1 K^j_{\overline Y}$, $D^j_1 S^j_{\overline X}$, $D^j_2 S^j_{\overline
X}$ and make $\{D^j_1$, $D^j_2$, $D^j_3$, $S^j_Z$, $S^j_{\overline Y}$, $L^j_Y$,
$K^j_X\}$ into a complete graph
\end{enumerate}

\end{enumerate}

\begin{lemma}\label{lem:sigma1}
$G'_\sigma$ is a subgraph of every chordal sandwich of $(G_\sigma,{\rm
forb}({\cal Q}_I))$.
\end{lemma}

\begin{proof}
Let $G'$ be a chordal sandwich of $(G_\sigma,{\rm forb}({\cal Q}_I))$.  We prove
the claim by showing that $G'$ contains all edges defined in
(\ref{enum:2})-(\ref{enum:6}).

For (\ref{enum:2}), let us consider true shoulders $S^j_W$, $S^{j'}_{W'}$ and
true knees $K^j_W$, $K^{j'}_{W'}$ and $L^j_W$, $L^{j'}_{W'}$ (if they exist). We
allow that $W$ is possibly equal to $W'$ and possibly $j=j'$. First, we observe
that, by (\ref{enum:1}), the true knees $K^j_W$ and $K^{j'}_{W'}$ are adjacent
to $B$.  Therefore, by Lemma \ref{lem:star}, the vertices $S^j_W$, $K^j_W$,
$L^j_W$ are adjacent to $H_W$ and $F^j$, whereas $S^{j'}_{W'}$, $K^{j'}_{W'}$,
$L^{j'}_{W'}$ are adjacent to $H_{W'}$ and $F^{j'}$. Also, $H_W$ is adjacent to
$H_{W'}$ and $F^j$ is adjacent to $F^{j'}$, since $\delta\in H_W\cap H_{W'}$ and
$\mu\in F^j\cap F^{j'}$, respectively.  Further, $H_W F^j$, $H_W F^{j'}$,
$H_{W'} F^j$, $H_{W'} F^{j'}$ are not edges of $G'$, since $H_W|F^j$,
$H_W|F^{j'}$, $H_{W'}|F^j$, $H_{W'}|F^{j'}$ are in ${\cal Q}_I$.  Thus, if
$j=j'$ and $W$ is equal to $W'$, then, by Lemma~\ref{lem:4cycle} applied to
cycles $\{u$,\,$H_W$,\,$u'$,\,$F^j\}$ where $u,u'\in\{S^j_W$, $S^{j'}_{W'}$,
$K^j_W$, $K^{j'}_{W'}$, $L^j_W$, $L^{j'}_{W'}\}$, we conclude that $\{S^j_W$,
$S^{j'}_{W'}$, $K^j_W$, $K^{j'}_{W'}$, $L^j_W$, $L^{j'}_{W'}\}$ forms a complete
graph in $G'$.  If $j\neq j'$ and $W$ is not equal to $W'$, we reach the same
conclusion by Lemma~\ref{lem:6cycle} applied to the cycles
$\{u$,\,$H_W$,\,$H_{W'}$,\,$u'$,\,$F^{j'}$,\,$F^j\}$. Otherwise, we obtain the
conclusion by applying Lemma~\ref{lem:5cycle} either to cycles
$\{u$,\,$H_W$,\,$u'$,\,$F^{j'}$,\,$F^j\}$ or cycles
$\{u$,\,$F^j$,\,$u'$,\,$H_{W'}$,\,$H_W\}$.~This~proves~(\ref{enum:2}).

For (\ref{enum:3}), consider the vertex $A_i$ for $i\in\{1\ldots n\}$.  Let
$W\in\{v_i,\overline{v_i}\}$ be such that $W=1$.  Then, for each $j\in\Delta_i$,
the vertex $K^j_W$ is adjacent to $B$ by (\ref{enum:1}). Thus, by Lemma
\ref{lem:star}, both $K^j_W$ and $L^j_W$ (if exists) are adjacent to $A_i$.
This proves (\ref{enum:3}).

For (\ref{enum:4}), we consider $1\leq i'\leq i\leq n$.  Let
$W'\in\{v_{i'},\overline{v_{i'}}\}$ be such that $W'=1$.  Then, for all
$j\in\Delta_{i'}$, the vertex $K^j_{W'}$ is adjacent to $B$ by (\ref{enum:1}),
and hence, the vertices $S^j_{W'}$, $K^j_{W'}$ and $L^j_{W'}$ (if exists) are
adjacent by Lemma \ref{lem:star} to $H_{v_{i'}}$, $H_{\overline{v_{i'}}}$.  In
other words, the vertices $H_{v_{i'}}$, $H_{\overline{v_{i'}}}$ are adjacent to
all true knees and true shoulders of the literals $v_{i'}$, $\overline{v_{i'}}$.
Thus, we may assume that~$i'<i$. Now, the vertex $H_{v_{i'}}$ is adjacent to
$H_{v_i}$,$H_{\overline{v_i}}$, since $\delta\in H_{v_i}\cap
H_{\overline{v_i}}\cap H_{v_{i'}}$.  Let $W\in\{v_i,\overline{v_i}\}$ be such
that $W=1$.  Then $K^j_W$ is adjacent to $B$ by (\ref{enum:1}), and hence,
$S^j_W$ is adjacent to $H_{v_i}$, $H_{\overline {v_i}}$ by Lemma \ref{lem:star}.
Also, $S^j_W$ is adjacent to all true knees and true shoulders of the literals
$v_{i'}$, $\overline{v_{i'}}$, by (\ref{enum:2}).  Further, the vertex $S^j_W$
is not adjacent to $H_{v_{i'}}$, since $H_{v_{i'}}|S^j_W$ is in ${\cal Q}_I$.
Thus, by Lemma~\ref{lem:4cycle}, both $H_{v_i}$ and $H_{\overline{v_i}}$ are
adjacent to all true knees and true shoulders of the literals $v_{i'}$,
$\overline{v_{i'}}$.  This proves (\ref{enum:4}).

For (\ref{enum:5}), consider $1\leq j\leq j'\leq m$.  Again, we observe that if
$K^{j'}_W$ is a true knee, then $K^{j'}_W$ is adjacent to $B$ by (\ref{enum:1}),
and hence, $S^{j'}_W$, $K^{j'}_W$, and $L^{j'}_W$ (if exists) are adjacent to
$F^{j'}$ by Lemma \ref{lem:star}. In other words, the vertex $F^{j'}$ is
adjacent to all true knees and true shoulders of the clause ${\cal C}_{j'}$. So,
we may assume that $j<j'$. Now, let $K^j_W$ be any true knee of the clause
${\cal C}_j$.  Then $K^j_W$ is adjacent to $B$, and hence, to $F^j$ by
(\ref{enum:1}) and Lemma \ref{lem:star}, respectively.  Also, $K^j_W$ is
adjacent to all true shoulders and true knees of ${\cal C}_{j'}$ by
(\ref{enum:2}).  Further, $F^j$ is adjacent to $F^{j'}$, since $\mu\in F^j\cap
F^{j'}$, and the vertex $K^j_W$ is not adjacent to $F^{j'}$, since
$K^j_W|F^{j'}$ is in ${\cal Q}_I$. Thus, by Lemma~\ref{lem:4cycle}, the vertex
$F^j$ is adjacent to all true knees and true shoulders of the clause ${\cal
C}_{j'}$. This~proves~(\ref{enum:5}).

For (\ref{enum:6}), let $i\in\{1\ldots n\}$ and consider $j,j'\in\Delta_i$ with
$j\leq j'$. Let $W\in\{v_i,\overline{v_i}\}$ be such that $W=1$.  Observe that
$K^j_W$ is adjacent to $S^j_{\overline W}$, since $\beta^j_{\overline W}\in
S^j_{\overline W}\cap K^j_W$. If $L^j_W$ exists, also $L^j_W$ is adjacent to
$S^j_{\overline W}$, since then $\beta^j_{\overline W}\in S^j_{\overline W}\cap
L^j_W$.  Thus, we may assume that $j<j'$.  Now, $S^{j'}_{\overline W}$ is
adjacent to $S^j_{\overline W}$ and $K^{j'}_W$, since $\alpha_{\overline W}\in
S^j_{\overline W}\cap S^{j'}_{\overline W}$, and $\beta^{j'}_{\overline W}\in
S^{j'}_{\overline W}\cap K^{j'}_W$.  Also, $K^j_W$ and $L^j_W$ (if exists) are
adjacent to $K^{j'}_W$ by (\ref{enum:2}). Further, $S^j_{\overline W} K^{j'}_W$
is not an edge of $G'$, since $S^j_{\overline W}|K^{j'}_W$ is in ${\cal Q}_I$.
Thus, by Lemma~\ref{lem:4cycle}, the vertices $K^j_W$, $L^j_W$ (if exists) are
adjacent to $S^{j'}_{\overline W}$. This proves (\ref{enum:6}).

The proof is now complete.
\end{proof}

\begin{lemma}\label{lem:sigma2}
If $\sigma$ is a satisfying assignment for $I$, then $G^*_\sigma$ is a subgraph
of every chordal sandwich of $(G_\sigma,{\rm forb}({\cal Q}_I))$.
\end{lemma}

\begin{proof}
Let $G'$ be a chordal sandwich of $(G_\sigma,{\rm forb}({\cal Q}_I))$, and
assume that $\sigma$ is a satisfying assignment for $I$. That is, in each clause
${\cal C}_j=X\vee Y\vee Z$, either $X=1$, $Y=Z=0$, or $Y=1$, $X=Z=0$, or $Z=1$,
$X=Y=0$. 

By Lemma \ref{lem:sigma1}, the graph $G'$ contain all edges defined in
(\ref{enum:2})-(\ref{enum:6}). Thus it remains to prove that it also contains
the edges defined in (\ref{enum:7}).

Consider $j\in\{1\ldots m\}$ where ${\cal C}_j=X\vee Y\vee Z$.  By the
rotational symmetry between $X$, $Y$, and $Z$, we may assume that $X=1$, $Y=0$,
and $Z=0$.  Observe that $K^j_Z$ is adjacent to $K^j_X$ and $L^j_Z$, since
$\lambda^j\in K^j_Z\cap K^j_X$ and $\beta^j_{\overline Z}\in K^j_Z\cap L^j_Z$.
Further, $K^j_{\overline X}$ is adjacent to $L^j_Z$ and $S^j_X$, since
$\gamma^j_1\in L^j_Z\cap K^j_X$ and $\beta^j_X\in K^j_{\overline X}\cap S^j_X$.
By (\ref{enum:2}), also $K^j_X$ is adjacent to $S^j_X$.  Moreover, $S^j_X K^j_Z$
and $K^j_X K^j_{\overline X}$ are not edges of $G'$, since $S^j_X|K^j_Z$,
$K^j_X|K^j_{\overline X}$ are in ${\cal Q}_I$.  Thus, by Lemma~\ref{lem:5cycle}
applied to the cycle $\{L^j_Z$, $K^j_Z$, $K^j_X$, $S^j_X$, $K^j_{\overline
X}\}$, we conclude that $L^j_Z$ is adjacent to $S^j_X$ and $K^j_X$.  Now,
observe that $L^j_Y$ is adjacent to $K^j_Y$ and $K^j_{\overline Z}$, since
$\beta^j_{\overline Y}\in L^j_Y\cap K^j_Y$ and $\gamma^j_3\in L^j_Y\cap
K^j_{\overline Z}$. Recall that $K^j_Z$ is adjacent to $L^j_Z$ and also to
$K^j_Y$, since $\lambda^j\in K^j_Z\cap K^j_Y$.  Moreover, $S^j_X$ is adjacent to
$K^j_{\overline Z}$ and $L^j_Z$ by (\ref{enum:2}) and the above.  Further,
$K^j_{\overline Z} L^j_Z$, $S^j_X L^j_Y$, $S^j_X K^j_Z$ are not edges of $G'$,
since $K^j_{\overline Z}|L^j_Z$, $S^j_X|L^j_Y$, $S^j_X|K^j_Z$ are in ${\cal
Q}_I$.  Thus, by Lemma~\ref{lem:6cycle} applied to the cycle $\{K^j_Y$, $L^j_Y$,
$K^j_{\overline Z}$, $S^j_X$, $L^j_Z$, $K^j_Z\}$, we conclude that $K^j_Y$ is
adjacent to $K^j_{\overline Z}$, $S^j_X$, and $L^j_Z$.  Next, observe that
$S^j_{\overline Z}$ is adjacent to $K^j_{\overline Z}$ and $K^j_Z$ by
(\ref{enum:2}) and since $\beta^j_{\overline Z}\in S^j_{\overline Z}\cap K^j_Z$.
Recall that $K^j_Y$ is adjacent to $K^j_{\overline Z}$ and $K^j_Z$. Further,
$K^j_Z K^j_{\overline Z}$ is not an edge of $G'$, since $K^j_Z|K^j_{\overline
Z}$ is in ${\cal Q}_I$.  Thus, by Lemma~\ref{lem:4cycle}, the vertex
$S^j_{\overline Z}$ is adjacent to $K^j_Y$.  Now, recall that $L^j_Z$ is
adjacent to  $S^j_X$ and $K^j_Z$, and $S^j_X K^j_Z$ is not an edge of $G'$.
Also, $F^j$ is adjacent to $S^j_X$ and $K^j_Z$ by (\ref{enum:5}) and since
$\lambda^j\in F^j\cap K^j_Z$. Thus, by Lemma~\ref{lem:4cycle}, the vertex
$L^j_Z$ is adjacent to $F^j$.  Now, observe that $D^j_1$ is adjacent to $K^j_X$,
$K^j_{\overline X}$, since $\lambda^j\in D^j_1\cap K^j_X$ and $\gamma^j_1\in
D^j_1\cap K^j_{\overline X}$.  Recall that also $S_X$ is adjacent to both
$K^j_X$ and $K^j_{\overline X}$, and that $K^j_X K^j_{\overline X}$ is not an
edge of $G'$. Thus, by Lemma~\ref{lem:4cycle}, we have that $D^j_1$ is adjacent
to $S^j_X$.  Next, observe that $D^j_2$ is adjacent to $K^j_Y$, $K^j_{\overline
Y}$, since $\lambda^j\in D^j_2\cap K^j_Y$ and $\gamma^j_2\in D^j_2\cap
K^j_{\overline Y}$.  Recall that $K^j_Y$ is adjacent to $K^j_{\overline Z}$ and
$S^j_X$. Also, $K^j_{\overline Y}$ is adjacent to $S^j_X$, $S^j_{\overline Y}$,
$K^j_{\overline Z}$ by (\ref{enum:2}), and $K^j_Y$ is adjacent to $S^j_{\overline
Y}$, since $\beta^j_{\overline Y}\in K^j_Y\cap S^j_{\overline Y}$.  Further,
$K^j_Y K^j_{\overline Y}$ is not an edge of $G'$, since $K^j_Y|K^j_{\overline
Y}$ is in ${\cal Q}_I$. Thus, by Lemma~\ref{lem:4cycle}, the vertices $S^j_X$,
$S^j_{\overline Y}$, $K^j_{\overline Z}$ are adjacent to $D^j_2$.  Now, observe
that $D^j_1$, $D^j_2$ are adjacent to $K^j_Z$, since $\lambda^j\in D^j_1\cap
D^j_2\cap K^j_Z$. Also, recall that $S^j_X$ is adjacent to $D^j_1$, $D^j_2$,
$L^j_Z$, the vertex $K^j_Z$ is adjacent to $S^j_{\overline Z}$, $L^j_Z$, and
$S^j_X K^j_Z$ is not an edge of $G'$.  Further, $S^j_X$ is adjacent to
$S^j_{\overline Z}$ by (\ref{enum:2}).  Thus, by Lemma~\ref{lem:4cycle}, both
$D^j_1$ and $D^j_2$ are adjacent to $S^j_{\overline Z}$ and $L^j_Z$.
Next, observe that $D^j_3$ is adjacent to $K^j_Z$, $K^j_{\overline Z}$, since
$\lambda^j\in D^j_3\cap K^j_Z$ and $\gamma^j_3\in D^j_3\cap K^j_{\overline Z}$.
Recall that also $S^j_{\overline Z}$ is adjacent to $K^j_Z$, $K^j_{\overline
Z}$, and that $K^j_Z K^j_{\overline Z}$ is not an edge of $G'$.  Thus, by
Lemma~\ref{lem:4cycle}, the vertex $D^j_3$ is adjacent to $S^j_{\overline Z}$.
Further, recall that $L^j_Z$ is adjacent to $K^j_Z$, $S^j_X$, the vertex
$K^j_{\overline Z}$ is adjacent to $S^j_X$, and $S^j_X K^j_Z$ and
$K^j_{\overline Z} L^j_Z$ are not edges of $G'$.  Thus, Lemma~\ref{lem:5cycle}
applied to $\{D^j_3$, $K^j_Z$, $L^j_Z$, $S^j_X$, $K^j_{\overline Z}\}$ yields
that $D^j_3$ is adjacent to both $L^j_Z$ and $S^j_X$. Moveover, $S^j_{\overline
Y}$ is also adjacent to $S^j_X$ by (\ref{enum:2}), and $L^j_Y$ is also adjacent
to $D^j_3$, $S^j_{\overline Y}$, since $\gamma^j_3\in D^j_3\cap L^j_Y$ and
$\beta^j_{\overline Y}\in S^j_{\overline Y}\cap L^j_Y$. Further, recall that
$S^j_X L^j_Y$ is not an edge of $G'$. Thus, by Lemma~\ref{lem:4cycle} applied to
$\{D^j_3$, $L^j_Y$, $S^j_{\overline Y}$, $S^j_X\}$, the vertex $D^j_3$ is
adjacent to $S^j_{\overline Y}$.

To prove (\ref{enum:7}), we observe that the above analysis yields that $G'$
contains edges $F^j L^j_Z$, $K^j_X L^j_Z$, $K^j_Y K^j_{\overline Z}$, $D^j_2
K^j_{\overline Z}$, $D^j_2 S^j_{\overline Y}$, and $D^j_3 S^j_{\overline Y}$.
It remains to show that $\{D^j_1$,~$D^j_2$, $D^j_3$, $S^j_X$, $S^j_{\overline
Z}$, $L^j_Z$, $K^j_Y\}$ forms a complete graph.  By the above paragraph, we have
that $S^j_X$, $S^j_{\overline Z}$, $L^j_Z$ are adjacent to $D^j_1$, $D^j_2$,
$D^j_3$. Also, $D^j_1$, $D^j_2$, $D^j_3$ and $K^j_Y$ are pair-wise adjacent,
since $\lambda^j\in D^j_1\cap D^j_2\cap D^j_3\cap K^j_Y$. Further, $L^j_Z$ is
adjacent to $S^j_X$, and $K^j_Y$ is adjacent to $S^j_X$, $S^j_{\overline Z}$,
$L^j_Z$, by the above paragraph.  Finally, $S^j_{\overline Z}$ is adjacent to
$S^j_X$ and $L^j_Z$ by (\ref{enum:2}) and since $\beta^j_{\overline Z}\in
S^j_{\overline Z}\cap L^j_Z$.  This proves~(\ref{enum:7}).

The proof is now complete.
\end{proof}

\begin{lemma}\label{lem:sigma3}
If $\sigma$ is a satisfying assignment for $I$, then $G^*_\sigma$ is chordal.
\end{lemma}
\begin{proof}
Again, assume that $\sigma$ is a satisfying assignment for $I$. That is, for each
clause ${\cal C}_j=X\vee Y\vee Z$, either $X=1$, $Y=Z=0$, or $Y=1$, $X=Z=0$, or
$Z=1$, $X=Y=0$.
Consider the following partition $V_1\cup V_2\cup V_3\cup V_4\cup V_5$ of
$V(G^*_\sigma)$ where
$V_1=\{$false knees, $D$-vertices$\}$,
$V_2=\{$false shoulders$\}$,
$V_3=\{A$-vertices$\}$,
$V_4=\{H$-vertices, $F$-vertices$\}$, and
$V_5=\{$true knees, true shoulders,~the~vertex~$B\}$.

Let $\pi$ be an enumeration of $V(G^*_\sigma)$ constructed by listing the elements
of $V_1$, $V_2$, $V_3$, $V_4$, $V_5$ in that order such that:
\begin{enumerate}[($\bullet$)]
\item the elements of $V_1$ are listed by considering each clause ${\cal C}_j=X\vee
Y\vee Z$ and listing vertices (based on the truth assignment) as follows:

\begin{enumerate}[a)]
\item if $X=1$, then list $K^j_{\overline X}$, $K^j_Z$, $L^j_Y$, $L^j_Z$,
$D^j_1$, $K^j_Y$, $D^j_3$, $D^j_2$ in that order,

\item if $Y=1$, then list $K^j_{\overline Y}$, $K^j_X$, $L^j_Z$, $L^j_X$,
$D^j_2$, $K^j_Z$, $D^j_1$, $D^j_3$ in that order,

\item if $Z=1$, then list $K^j_{\overline Z}$, $K^j_Y$, $L^j_X$, $L^j_Y$,
$D^j_3$, $K^j_X$, $D^j_2$, $D^j_1$ in that order,
\end{enumerate}
\item the elements of $V_2$ (the false shoulders) are listed by listing the
false shoulders of the clauses ${\cal C}_1$, ${\cal C}_2$, \ldots, ${\cal C}_m$
in that order,

\item the elements of $V_4$ are listed as follows: first the vertices
$H_{v_1}$,$H_{\overline{v_1}}$, $H_{v_2}$,$H_{\overline{v_2}}$, \ldots
$H_{v_n}$,$H_{\overline{v_n}}$ in that order, then $F^m$, $F^{m-1}$, \ldots,
$F^1$ in that order,

\item the elements of $V_3$ and $V_5$ are listed in any order.
\end{enumerate}

\noindent We show that $\pi$ is a perfect elimination ordering of $G^*_\sigma$ which
implies the claim.

First, consider $V_1$. Let $j\in\{1\ldots m\}$ and let ${\cal C}_j=X\vee Y\vee
Z$. By the rotational symmetry between $X$, $Y$, $Z$, assume that $X=1$ and
$Y=Z=0$. So, $\pi$ lists the false knees and $D$-vertices of ${\cal C}_j$ as
$K^j_{\overline X}$, $K^j_Z$, $L^j_Y$, $L^j_Z$, $D^j_1$, $K^j_Y$, $D^j_3$,
$D^j_2$.

First, consider the vertex $K^j_{\overline X}$. Recall that $K^j_{\overline
X}=\{\beta^j_X, \gamma^j_1\}$. Observe that $S^j_X$ is the only other vertex
containing $\beta^j_X$, and $L^j_Z$, $D^j_1$ are the only other vertices
containing $\gamma^j_1$. Moreover, none of the rules
(\ref{enum:1})-(\ref{enum:7}) adds edges incident to $K^j_{\overline X}$.  Thus,
$S^j_X$, $L^j_Z$, $D^j_1$ are the only neighbours of $K^j_{\overline X}$, and
they are pair-wise adjacent by (\ref{enum:7}). This proves that $K^j_{\overline
X}$ is indeed simplicial in $G^*_\sigma$.

Next, consider $K^j_Z$. Since $K^j_Z=\{\beta^j_Z, \lambda^j\}$, we conclude that
$K^j_Z$ is adjacent to $S^j_{\overline Z}$, $L^j_Z$, $K^j_X$, $K^j_Y$, $D^1_j$,
$D^2_j$, $D^3_j$, and $F^j$. Moreover, $K^j_Z$ has no other neighbours by
observing the rules (\ref{enum:1})-(\ref{enum:7}).  Now, by (\ref{enum:7}), we
conclude that $S^j_{\overline Z}$, $L^j_Z$, $K^j_Y$, $D^j_1$, $D^j_2$, $D^j_3$
are pair-wise adjacent.  Also, the vertices $F^j$, $K^j_X$, $K^j_Y$, $D^j_1$,
$D^j_2$, $D^j_3$ are pair-wise adjacent, since they all contain $\lambda^j$.
Further, $F^j$ is adjacent to $S^j_{\overline Z}$ and $L^j_Z$ by (\ref{enum:5})
and (\ref{enum:7}), respectively, and $K^j_X$ is adjacent to $S^j_{\overline Z}$
and $L^j_Z$ by (\ref{enum:2}) and (\ref{enum:7}), respectively.  This proves
that $K^j_Z$ is simplicial in $G^*_\sigma$.

Now, consider $L^j_Y$. The neighbours of $L^j_Y$ are $S^j_{\overline Y}$,
$K^j_Y$, $K^j_{\overline Z}$, and $D^j_3$.  So, $S^j_{\overline Y}$ is adjacent
to $K^j_{\overline Z}$, $D^j_3$,  and $K^j_Y$  by (\ref{enum:2}),
(\ref{enum:7}), and since $\beta^j_{\overline Y}\in S^j_{\overline Y}\cap
K^j_Y$.  Similarly, $K^j_Y$ is adjacent to $K^j_{\overline Z}$ and $D^j_3$ by
(\ref{enum:7}) and since $\lambda^j\in K^j_Y\cap D^j_3$. Finally,
$K^j_{\overline Z}$ is adjacent to $D^j_3$, since $\gamma^j_3\in K^j_{\overline
Z}\cap D^j_3$. Thus $L^j_Y$ is simplicial in $G^*_\sigma$.

Next, consider $L^j_Z$. The neighbours of $L^j_Z$ are $F^j$, $K^j_X$, $K^j_Y$,
$K^j_Z$, $D^j_1$, $D^j_2$, $D^j_3$, $S^j_X$, $S^j_{\overline Z}$, and
$K^j_{\overline X}$. By (\ref{enum:7}), the vertices $D^j_1$, $D^j_2$, $D^j_3$,
$S^j_X$, $S^j_{\overline Z}$, $K^j_Y$ are pair-wise adjacent. Also, $F^j$,
$K^j_X$, $K^j_Y$, $D^j_1$, $D^j_2$, $D^j_3$ are pair-wise adjacent, since they
all contain $\lambda_j$.  Further,  $K^j_X$ and $F^j$ are adjacent to $S^j_X$,
$S^j_{\overline Z}$ by (\ref{enum:2}) and (\ref{enum:5}), respectively. This
proves that $L^j_Z$ is simplicial in $G^*_\sigma-\{K^j_{\overline X}, K^j_Z\}$.

Now, consider $D^j_1$. The neighbours of $D^j_1$ are $F^j$, $K^j_X$, $K^j_Y$,
$K^j_Z$, $D^j_2$, $D^j_3$, $S^j_X$, $S^j_{\overline Z}$, $L^j_Z$, and
$K^j_{\overline X}$. By (\ref{enum:7}), the vertices $D^j_2$, $D^j_3$, $S^j_X$,
$S^j_{\overline Z}$, $K^j_Y$ are pair-wise adjacent.  Also, $F^j$, $K^j_X$,
$K^j_Y$, $D^j_2$, $D^j_3$ are pair-wise adjacent, since they all
contain~$\lambda^j$.  Further, $K^j_X$ and $F^j$ are adjacent to $S^j_X$,
$S^j_{\overline Z}$ by (\ref{enum:2}) and (\ref{enum:5}), respectively. This
proves that $D^j_1$ is simplicial in $G^*_\sigma-\{K^j_{\overline X}, K^j_Z,
L^j_Z\}$.

Next, consider $K^j_Y$. The neighbours of $K^j_Y$ are $F^j$, $K^j_X$, $K^j_Z$,
$D^j_1$, $D^j_2$, $D^j_3$, $S^j_X$, $S^j_{\overline Y}$, $S^j_{\overline Z}$,
$K^j_{\overline Z}$, $L^j_Y$, and $L^j_Z$.  By (\ref{enum:7}), the vertices
$D^j_2$, $D^j_3$, $S^j_X$, $S^j_{\overline Z}$ are pair-wise adjacent. Also,
$F$, $K^j_X$, $D^j_2$, $D^j_3$ are pair-wise adjacent, since they all contain
$\lambda^j$. Further, by (\ref{enum:2}), the vertices $S^j_X$, $S^j_{\overline
Y}$, $S^j_{\overline Z}$, $K^j_X$, and $K^j_{\overline Z}$ are pair-wise
adjacent, and are adjacent to $F^j$ by (\ref{enum:5}).  Moreover, by
(\ref{enum:7}), both $S^j_{\overline Y}$ and $K^j_{\overline Z}$ are adjacent
$D^j_2$, and are also adjacent to $D^j_3$ by  (\ref{enum:7}) and since
$\gamma^j_3\in K^j_{\overline Z}\cap D^j_3$, respectively.  This proves that
$K^j_Y$ is simplicial in $G^*_\sigma-\{K^j_Z, L^j_Y, L^j_Z, D^j_1\}$.

Now, consider $D^3_j$. The neighbours of $D^3_j$ are $F^j$, $K^j_X$, $K^j_Y$,
$K^j_Z$, $D^j_1$, $D^j_2$, $S^j_X$, $S^j_{\overline Y}$, $S^j_{\overline Z}$,
$K^j_{\overline Z}$, $L^j_Z$, and $L^j_Y$.  By (\ref{enum:2}), the vertices
$S^j_X$, $S^j_{\overline Y}$, $S^j_{\overline Z}$, $K^j_X$, $K^j_{\overline Z}$ are
pair-wise adjacent. Also, $F^j$, $K^j_X$, $D^j_2$ are pair-wise adjacent, since
they all contain $\lambda^j$. Further, $F^j$ and $D^j_2$ are adjacent to 
$S^j_X$, $S^j_{\overline Y}$, $S^j_{\overline Z}$, $K^j_{\overline Z}$ by
(\ref{enum:5}) and (\ref{enum:7}), respectively. Thus $D^3_j$ is simplicial in
$G^*_\sigma-\{K^j_Z, L^j_Y, L^j_Z, D^j_1, K^j_Y\}$.

Finally, consider $D^2_j$. The neighbours of $D^2_j$ are $F^j$, $K^j_X$,
$K^j_Y$, $K^j_Z$, $D^1_j$, $D^3_j$, $S^j_X$, $S^j_{\overline Y}$,
$S^j_{\overline Z}$, $K^j_{\overline Z}$, $K^j_{\overline Y}$, $L^j_X$ and
$L^j_Z$. By (\ref{enum:2}), the vertices $S^j_X$, $S^j_{\overline Y}$,
$S^j_{\overline Z}$, $K^j_X$, $L^j_X$, $K^j_{\overline Y}$, $K^j_{\overline Z}$
are pair-wise adjacent, and are adjacent to $F$ by (\ref{enum:5}).  Thus $D^2_j$
is simplicial in $G^*_\sigma-\{K^j_Z, L^j_Z, D^j_1, K^j_Y, D^j_3\}$. This concludes
the vertices in $V_1$.

We now consider $V_2$. Let $j\in\{1\ldots m\}$ and consider a false shoulder
$S^j_W$ for some $W=0$. Let $i$ be such that $W=v_i$ or $W=\overline{v_i}$.
Then the neighbours of $S^j_W$ are the vertices $H_W$, $A_i$, and the elements
of the following sets:

${\cal S}^-=\{S^{j'}_W~|~j'\in\Delta_i~{\rm and}~j'<j\}$
\qquad
${\cal S}^+=\{S^{j'}_W~|~j'\in\Delta_i~{\rm and}~j<j'\}$

${\cal K}^-=\{K^{j'}_{\overline W}, L^{j'}_{\overline W}~{\rm (if~exists)}~|~
j'\in\Delta_i~{\rm and}~j'\leq j\}$

\noindent By (\ref{enum:2}), the elements of ${\cal K}^-$ are pair-wise
adjacent.  Similarly, the elements of $\{H_W$, $A_i\}\cup{\cal S}^+$ are pair-wise
adjacent, since they all contain $\alpha_W$.  Further, each element of ${\cal
S}^+$ is adjacent to every element of ${\cal K}^-$ by (\ref{enum:6}), and each
element of~${\cal K}^-$ is adjacent to $A_i$ and $H_W$ by (\ref{enum:3}) and
(\ref{enum:4}), respectively. This proves that $S^j_W$ is simplicial in
$G^*_{\sigma}-{\cal S}^-$.  Finally, note that the elements of ${\cal S}^-$ are false
shoulders in clauses ${\cal C}_1$, \ldots, ${\cal C}_{j-1}$.  This concludes
the elements of $V_2$.

For $V_3$, let $i\in\{1\ldots n\}$ and consider the vertex $A_i$.  The
neighbours of $A_i$ are the vertices $H_{v_i}$, $H_{\overline{v_i}}$, all
shoulders of the literals $v_i$, $\overline{v_i}$, and all true knees
of~$v_i$,\,$\overline{v_i}$. By (\ref{enum:2}), the true knees and true
shoulders of $v_i$, $\overline{v_i}$ are pair-wise adjacent, and are adjacent to
both $H_{v_i}$ and $H_{\overline{v_i}}$ by (\ref{enum:4}). Also, $H_{v_i}$ is
adjacent to $H_{\overline{v_i}}$, since $\delta\in H_{v_i}\cap
H_{\overline{v_i}}$. Thus $A_i$ is simplicial in $G^*_\sigma-V_2$. This
concludes $V_3$.

Now, we consider $V_4$. Let $i\in\{1\ldots n\}$ and consider $H_{v_i}$,
$H_{\overline{v_i}}$.  The vertices $H_{v_i}$, $H_{\overline{v_i}}$ are adjacent
to the vertices $B$, $A_i$, the elements of the following sets

${\cal H}^-=\{ H_{v_{i'}}, H_{\overline{v_{i'}}}~|~i'<i\}$
\qquad
${\cal H}^+=\{ H_{v_{i'}}, H_{\overline{v_{i'}}}~|~i<i'\}$

\noindent and all true knees, true shoulders of $v_{i'}$, $\overline{v_{i'}}$
for all $i'\in\{1\ldots i\}$. Further, $H_{v_i}$ is adjacent to
$H_{\overline{v_i}}$, to all shoulders of $v_i$ and to no other vertices,
whereas $H_{\overline{v_i}}$ is adjacent $H_{v_i}$, to all shoulders of
$\overline{v_i}$ and to no other vertices.  Now, by (\ref{enum:2}), the true
knees and true shoulders of $v_{i'}$, $\overline{v_{i'}}$ for all
$i'\in\{1\ldots i\}$, are pair-wise adjacent, and are adjacent to $B$ and each
element of ${\cal H}^+$ by (\ref{enum:1}) and (\ref{enum:4}), respectively.
Also, the elements of $\{B\}\cup{\cal H}^+$ are pair-wise adjacent, since they
all contain $\delta$.  Finally, observe that the false shoulders of $v_i$,
$\overline{v_i}$ belong to $V_2$. This proves that both $H_{v_i}$ and
$H_{\overline{v_i}}$ are simplicial in $G^*_\sigma-(V_2\cup V_3\cup {\cal H}^-)$
as required.

Next, let $j\in\{1\ldots m\}$ and consider $F^j$. Let ${\cal C}_j=X\vee Y\vee
Z$, and by the rotational symmetry, assume that $X=1$ and $Y=Z=0$.  Then the
neighbours of $F^j$ are $B$, $K^j_Y$, $K^j_Z$, $D^j_1$, $D^j_2$,
$D^j_3$, $L^j_Z$, the elements of the following sets

${\cal F}^-=\{F^{j'}~|~j'<j\}$\qquad ${\cal F}^+=\{F^{j'}~|~j<j'\}$

\noindent and all true knees and true shoulders of the clause ${\cal C}_{j'}$
for all $j'\in\{j\ldots m\}$.  By~(\ref{enum:2}), the true knees and true
shoulders of the clause ${\cal C}_{j'}$ for all $j'\in\{j\ldots m\}$, are
pair-wise adjacent, and are adjacent to $B$ and each elements of ${\cal F}^-$ by
(\ref{enum:1}) and (\ref{enum:5}), respectively.
Also, the vertices of $\{B\}\cup{\cal F}^-$ are pair-wise adjacent, since they
all contain $\mu$. Thus $F^j$ is simplicial in $G^*_\sigma-(V_1\cup{\cal
F}^+)$. This concludes $V_4$.

Finally, observe that all vertices of $V_5$ are pair-wise adjacent by
(\ref{enum:1}) and (\ref{enum:2}).
That concludes the proof.
\end{proof}

\begin{lemma}\label{lem:sigma4}
For every chordal sandwich $G'$ of $({\rm int^*}({\cal Q}_I),{\rm forb}({\cal
Q}_I))$, there is $\sigma$ such that $G_\sigma$ is a subgraph of $G'$, and
such that $\sigma$ is a satisfying assignment~for~$I$.
\end{lemma}
\begin{proof}
By Lemma \ref{lem:star}, for each $i\in\{1\ldots n\}$, there is
$W\in\{v_i,\overline{v_i}\}$ such that for all $j\in\Delta_i$, the vertices
$S^j_W$, $K^j_W$, and $L^j_W$(if exists) are adjacent to $B$. Set
$\sigma(v_i)=1$ if $W=v_i$, and otherwise set $\sigma(v_i)=0$.  For such a
mapping $\sigma$, the graph $G'$ clearly contains all edges of $G_\sigma$. Thus,
by Lemma \ref{lem:sigma2}, the graph $G'_\sigma$ is a subgraph of $G'$, that is,
$G'$ contains the edges defined in~(\ref{enum:2})-(\ref{enum:6}).

It remains to prove that $\sigma$ is a satisfying assignment for $I$.  Let
$j\in\{1\ldots m\}$ and consider the clause ${\cal C}_j=X\vee Y\vee Z$.  If
$X=Y=1$, then the vertex $S^j_Y$ is a true shoulder, and $K^j_X$ is a true knee.
Thus, by (\ref{enum:2}), we conclude that $S^j_Y$ is adjacent $K^j_X$. However,
this is impossible, since $S^j_Y|K^j_X$ is in ${\cal Q}_Y$.  Similarly, if
$X=Z=1$, we have that $S^j_X$ is adjacent to $K^j_Z$ by (\ref{enum:2}) while
$S^j_X|K^j_Z$ is in ${\cal Q}_I$, and if $Y=Z=1$, then $S^j_Z$ is adjacent to
$K^j_Y$ by (\ref{enum:2}) while $S^j_Z|K^j_Y$ is in ${\cal Q}_I$.

Now, suppose that $X=Y=Z=0$. First, observe that $K^j_X$ is adjacent to $L^j_X$,
$K^j_Z$, and the vertex $L^j_Z$ is adjacent to $K^j_Z$, $K^j_{\overline X}$,
since $\beta^j_{\overline X}\in K^j_X\cap L^j_X$, $\lambda^j\in K^j_X\cap
K^j_Z$, $\beta^j_{\overline Z}\in L^j_Z\cap K^j_Z$, and $\gamma^j_1\in L^j_Z\cap
K^j_{\overline X}$.  Also, $K^j_{\overline X}$ is adjacent to $K^j_{\overline
Z}$ by (\ref{enum:2}).  Further, $K^j_{\overline Z} K^j_Z$, $K^j_{\overline Z}
L^j_Z$ and $K^j_{\overline X}L^j_X$ are not edges of $G'$, since $K^j_{\overline
Z}|K^j_Z$, $K^j_{\overline Z}|L^j_Z$, and $K^j_{\overline X}|L^j_X$ and in
${\cal Q}_I$.  Thus, if $L^j_X$ is adjacent to $K^j_{\overline Z}$, then by
Lemma~\ref{lem:6cycle} applied to $\{K^j_X$, $L^j_X$, $K^j_{\overline Z}$,
$K^j_{\overline X}$, $L^j_Z$, $K^j_Z\}$, we conclude that $K^j_X$ is adjacent to
$K^j_{\overline X}$, which is impossible since $K^j_{\overline X}|K^j_X$ is in
${\cal Q}_I$. Similarly, if $K^j_X$ is adjacent to $K^j_{\overline Z}$, then by
Lemma~\ref{lem:5cycle} applied to $\{K^j_X$, $K^j_{\overline Z}$,
$K^j_{\overline X}$, $L^j_Z$, $K^j_Z\}$, we again conclude that $K^j_X$ is
adjacent to $K^j_{\overline X}$, a contradiction.  So, we may assume that both
$K^j_X$ and $L^j_X$ are not adjacent to $K^j_{\overline Z}$.  Now, observe that
$L^j_Y$ is adjacent to $K^j_{\overline Z}$, $K^j_Y$, and the vertex $K^j_X$ is
adjacent to $L^j_X$, $K^j_Y$, since $\gamma^j_3\in K^j_{\overline Z}\cap L^j_Y$,
$\beta^j_{\overline Y}\in L^j_Y\cap K^j_Y$, $\beta^j_{\overline X}\in K^j_X\cap
L^j_X$, and $\lambda^j\in K^j_Y\cap K^j_X$.  Also, $K^j_{\overline Y}$ is
adjacent to $K^j_{\overline Z}$ and $L^j_X$ by (\ref{enum:2}) and since
$\gamma^j_2\in K^j_{\overline Y}\cap L^j_X$. Further, $K^j_{\overline Y} K^j_Y$
and $K^j_{\overline Y}L^j_Y$ are not edges of $G'$, since $K^j_{\overline
Y}|K^j_Y$ and $K^j_{\overline Y}|L^j_Y$ are in ${\cal Q}_I$. Recall that $K^j_X$
and $L^j_X$ are not adjacent to $K^j_{\overline Z}$.  Then this contradicts
Lemma \ref{lem:6cycle} when applied to $\{K^j_X$, $L^j_X$, $K^j_{\overline Y}$,
$K^j_{\overline Z}$, $L^j_Y$, $K^j_Y\}$.

Thus, it is not the case that $X=Y=Z=0$, and by the above also not $X=Y=1$, nor
$X=Z=1$, nor $Y=Z=1$. Therefore, either $X=1$, $Y=Z=0$, or $Y=1$, $X=Z=0$, or
$Z=1$, $X=Y=0$. This proves that $\sigma$ is indeed a satisfying assignment for
$I$, which concludes the proof.
\end{proof}

We are finally ready to prove Theorem \ref{thm:one}.\medskip

\begin{proofof}{Theorem \ref{thm:one}}
Let $G'$ be a minimal chordal sandwich of $({\rm int^*}({\cal Q}_I)$, ${\rm
forb}({\cal Q}_I))$.  By Lemma \ref{lem:sigma4}, there exists $\sigma$, a
satisfying assignment  for $I$, such that $G_{\sigma}$ is a subgraph fo $G'$.
Thus, $G'$ is also a chordal sandwich of $(G_{\sigma},{\rm forb}({\cal Q}_I))$,
and hence, $G^*_\sigma$ is a subgraph of $G'$ by Lemma \ref{lem:sigma2}.  But by
Lemma \ref{lem:sigma3}, $G^*_\sigma$ is chordal, and so $G'$ is isomorphic to
$G^*_\sigma$ by the minimality of $G'$.

Conversely, if $\sigma$ is a satisfying assignment for $I$, then the graph
$G^*_\sigma$ is chordal by Lemma \ref{lem:sigma3}. Moreover,  ${\rm int^*}({\cal
Q}_I)$ is a subgraph of $G^*_\sigma$, by definition, and $G^*_\sigma$ contains
no edges of ${\rm forb}({\cal Q}_I)$, also by definition.  Thus, $G^*_\sigma$ is
a chordal sandwich of $({\rm int^*}({\cal Q}_I),{\rm forb}({\cal Q}_I))$, and it
is minimal by Lemma \ref{lem:sigma2}.

This proves that by mapping each satisfying assigment $\sigma$ to the graph
$G^*_\sigma$, we obtain the required bijection. That concludes the proof.
\end{proofof}

Finally, we have all the pieces to prove Theorem \ref{thm:unique-phyl}.\medskip

\section{Proof of Theorem \ref{thm:unique-phyl}}\label{sec:proof-unique-phyl}

Consider an instance $I$ to {\sc one-in-three-3sat} and a satisfying assignment
for $I$. We construct  the collection ${\cal Q}_I$ of quartet trees, as well as
the ternary phylogenetic tree ${\cal T}_\sigma$ as described in Sections
\ref{sec:constr} and \ref{sec:unique-trees}, respectively.  Clearly,
constructing ${\cal Q}_I$ and ${\cal T}_\sigma$ takes polynomial time.  By
combining Theorem \ref{thm:sandwich} with Theorems \ref{thm:one} and
\ref{thm:unique-trees}, we obtain that $\sigma$ is the unique satisfying
assignment of $I$ if and only if ${\cal T}_\sigma$ is the only phylogenetic tree
that displays ${\cal Q}_I$. Since, by Theorem \ref{thm:unique-3sat}, it is
$NP$-hard to determine if an instance to {\sc one-in-three-3sat} has a unique
satisfying assignment, it is therefore $NP$-hard to decide, for a given
phylogenetic tree ${\cal T}$ and a collection of quartet trees ${\cal Q}$,
whether or not ${\cal Q}$ defines ${\cal T}$. 

That concludes the proof.

\bibliographystyle{acm}
\bibliography{phyl2}

\begin{thebibliography}{10}

\bibitem{agarwala1994}
{\sc Agarwala, R., and Fern\'{a}ndez-Baca, D.}
\newblock A polynomial-time algorithm for the perfect phylogeny problem when
  the number of character states is fixed.
\newblock {\em SIAM Journal of Computing 23\/} (1994), 1216--1224.

\bibitem{twostrikes}
{\sc Bodlaender, H.~L., Fellows, M.~R., and Warnow, T.~J.}
\newblock Two strikes against perfect phylogeny.
\newblock In {\em Proceedings of 19th International Colloquium on Automata,
  Languages and Programming, Lecture Notes in Computer Science 623\/} (1992),
  Springer Berlin/Heidelberg, pp.~273--283.

\bibitem{buneman}
{\sc Buneman, P.}
\newblock A characterization of rigid circuit graphs.
\newblock {\em Discrete Mathematics 9\/} (1974), 205--212.

\bibitem{60s1}
{\sc Camin, J., and Sokal, R.}
\newblock A method for deducing branching sequences in phylogeny.
\newblock {\em Evolution 19\/} (1965), 311--326.

\bibitem{counting-sat}
{\sc Creignou, N., and Hermann, M.}
\newblock Complexity of generalized satisfiability counting problems.
\newblock {\em Information and Computation 125\/} (1996), 1--12.

\bibitem{sandwich-strongly}
{\sc de~Figueiredo, C. M.~H., Faria, L., Klein, S., and Sritharan, R.}
\newblock On the complexity of the sandwich problems for strongly chordal
  graphs and chordal bipartite graphs.
\newblock {\em Theoretical Computer Science 381\/} (2007), 57--67.

\bibitem{dekker}
{\sc Dekker, M. C.~H.}
\newblock Reconstruction methods for derivation trees.
\newblock Master's thesis, Vrije Universiteit, Amsterdam, 1986.

\bibitem{dirac}
{\sc Dirac, G.~A.}
\newblock On rigid circuit graphs.
\newblock {\em Abhandlungen aus dem Mathematischen Seminar der Universit\"at
  Hamburg 25\/} (1961), 71--76.

\bibitem{70s1}
{\sc Estabrook, G.~F.}
\newblock Cladistic methodology: a discussion of the theoretical basis for the
  induction of evolutionary history.
\newblock {\em Annual Review of Ecology and Systematics 3\/} (1972), 427--456.

\bibitem{70s2}
{\sc Estabrook, G.~F., C.~S.~Johnson, J., and McMorris, F.~R.}
\newblock An idealized concept of the true cladistic character.
\newblock {\em Mathematical Biosciences 23\/} (1975), 263--272.

\bibitem{70s3}
{\sc Estabrook, G.~F., C.~S.~Johnson, J., and McMorris, F.~R.}
\newblock An algebraic analysis of cladistic characters.
\newblock {\em Discrete Mathematics 16\/} (1976), 141--147.

\bibitem{70s4}
{\sc Estabrook, G.~F., C.~S.~Johnson, J., and McMorris, F.~R.}
\newblock A mathematical foundation for the analysis of cladistic character
  compatibility.
\newblock {\em Mathematical Biosciences 29\/} (1976), 181--187.

\bibitem{sandwich}
{\sc Golumbic, M.~C., Kaplan, H., and Shamir, R.}
\newblock Graph sandwich problems.
\newblock {\em Journal of Algorithms 19\/} (1995), 449--473.

\bibitem{gordon86}
{\sc Gordon, A.~D.}
\newblock Consensus supertrees: The synthesis of rooted trees containing
  overlapping sets of labeled leaves.
\newblock {\em Journal of Classification 3\/} (1986), 335--348.

\bibitem{gusfield}
{\sc Gusfield, D.}
\newblock Efficient algorithms for inferring evolutionary trees.
\newblock {\em Networks 21\/} (1991), 19--28.

\bibitem{unique}
{\sc Juban, L.}
\newblock Dichotomy theorem for the generalized unique satisfiability problem.
\newblock In {\em Proceedings of the 12th International Symposium on
  Fundamentals of Computation Theory (FCT 99), Lecture Notes in Computer
  Science 1684\/} (1999), Springer Berlin/Heidelberg, pp.~327--337.

\bibitem{warnow1990}
{\sc Kannan, S.~K., and Warnow, T.~J.}
\newblock Triangulating 3-colored graphs.
\newblock {\em SIAM Journal on Discrete Mathematics 5\/} (1992), 249--258.

\bibitem{threestate}
{\sc Lam, F., Gusfield, D., and Sridhar, S.}
\newblock Generalizing the splits equivalence theorem and four gamete
  condition: Perfect phylogeny on three state characters.
\newblock In {\em Algorithms in Bioinformatics (WABI 2009), Lecture Notes in
  Computer Science 5724\/} (2009), Springer Berlin/Heidelberg, pp.~206--219.

\bibitem{60s3}
{\sc LeQuesne, W.~J.}
\newblock Further studies on the uniquely derived character concept.
\newblock {\em Systematic Zoology 21\/} (1972), 281--288.

\bibitem{60s4}
{\sc LeQuesne, W.~J.}
\newblock The uniquely evolved character concept and its cladistic application.
\newblock {\em Systematic Zoology 23\/} (1974), 513--517.

\bibitem{60s5}
{\sc LeQuesne, W.~J.}
\newblock The uniquely evolved character concept.
\newblock {\em Systematic Zoology 26\/} (1977), 218--223.

\bibitem{morris45}
{\sc McMorris, F.~R., Warnow, T., and Wimer, T.}
\newblock Triangulating vertex colored graphs.
\newblock {\em SIAM Journal on Discrete Mathematics 7\/} (1994), 296--306.

\bibitem{rosetarjan}
{\sc Rose, D., Tarjan, R., and Lueker, G.}
\newblock Algorithmic aspects of vertex elimination on graphs.
\newblock {\em SIAM Journal of Computing 5\/} (1976), 266--283.

\bibitem{semplesteel}
{\sc Semple, C., and Steel, M.}
\newblock A characterization for a set of partial partitions to define an
  {$X$}-tree.
\newblock {\em Discrete Mathematics 247\/} (2002), 169--186.

\bibitem{phyl-book}
{\sc Semple, C., and Steel, M.}
\newblock {\em Phylogenetics}.
\newblock Oxford lecture series in mathematics and its applications. Oxford
  University Press, 2003.

\bibitem{shaefer}
{\sc Shaefer, T.~J.}
\newblock The complexity of satisfiability problems.
\newblock In {\em Proceedings of 10th ACM Symposium on Theory of Computing
  (STOC)\/} (1978), pp.~216--226.

\bibitem{steel-web}
{\sc Steel, M.}
\newblock personal webpage, \url{http://www.math.canterbury.ac.nz/~m.steel/}.

\bibitem{steelnphard}
{\sc Steel, M.}
\newblock The complexity of reconstructing trees from qualitative characters
  and subtrees.
\newblock {\em Journal of Classification 9\/} (1992), 91--116.

\bibitem{west-book}
{\sc West, D.}
\newblock {\em Introduction to Graph Theory}.
\newblock Prentice Hall, 1996.

\bibitem{60s2}
{\sc Wilson, E.~O.}
\newblock A consistency test for phylogenies based upon contemporaneous
  species.
\newblock {\em Systematic Zoology 14\/} (1965), 214--220.

\end{thebibliography}

\end{document}